\documentclass[prd,aps,amsfonts,showpacs,longbibliography,notitlepage,nofootinbib,superscriptaddress]{revtex4-1}

\usepackage[utf8]{inputenc}
\usepackage[T1]{fontenc}
\usepackage{graphicx}
\usepackage[usenames,dvipsnames]{xcolor}
\usepackage{amsmath,amssymb,graphics,amsthm,isomath}
\usepackage{bbm}
\usepackage{bm}
\usepackage{braket}
\usepackage{stmaryrd}

\usepackage{comment}

\usepackage[colorlinks=true, urlcolor=violet, linkcolor=blue, citecolor=red, hyperindex=true, linktocpage=true]{hyperref}

\allowdisplaybreaks

\newtheorem{thm}{Theorem}
\numberwithin{thm}{section}
\newtheorem{cor}[thm]{Corollary}

\newtheorem{prop}[thm]{Proposition}

\newtheorem{defn}[thm]{Definition}
\newtheorem{conj}[thm]{Conjecture}

\renewcommand{\thesection}{\arabic{section}}
\renewcommand{\thesubsection}{\thesection.\arabic{subsection}}

\makeatletter
\renewcommand{\p@subsection}{}
\renewcommand{\p@subsubsection}{}
\makeatother

\usepackage{xcolor}
\usepackage{mathtools}
\usepackage[compat=1.0.0]{tikz-feynman}
\usepackage{subcaption}

\newcommand{\abs}[1]{\lvert{#1}\rvert}

\newcommand{\transpose}{\mathsf{T}}

\newcommand{\ident}[0]{\mathds{1}}

\usepackage{dsfont}

\usepackage{tikz}
\usetikzlibrary{quantikz2}

\usepackage[ruled,linesnumbered]{algorithm2e}

\SetCommentSty{mycommfont}
\SetKwInOut{Input}{Input}
\SetKwInOut{Return}{Return}

\usepackage[dvipsnames]{xcolor}
\usepackage{bm}
\usepackage{bbm}
\usepackage{multirow}
\usepackage{booktabs} 
\usepackage{caption}
\begin{document}

\title{Towards self-correcting quantum codes for neutral atom arrays}

\author{Jinkang Guo}
\affiliation{Department of Physics, University of Colorado, Boulder, CO 80309, USA}
\affiliation{Center for Theory of Quantum Matter, University of Colorado, Boulder, CO 80309, USA}

\author{Yifan Hong}
\affiliation{Department of Physics, University of Colorado, Boulder, CO 80309, USA}
\affiliation{Center for Theory of Quantum Matter, University of Colorado, Boulder, CO 80309, USA}
\affiliation{Joint Quantum Institute \& Joint Center for Quantum Information and Computer Science, NIST/University of Maryland, College Park, MD 20742, USA}

\author{Adam Kaufman}
\affiliation{Department of Physics, University of Colorado, Boulder, CO 80309, USA}
\affiliation{JILA, University of Colorado and National Institute of Standards and Technology, Boulder, Colorado 80309, USA}

\author{Andrew Lucas}
\email{andrew.j.lucas@colorado.edu}
\affiliation{Department of Physics, University of Colorado, Boulder, CO 80309, USA}
\affiliation{Center for Theory of Quantum Matter, University of Colorado, Boulder, CO 80309, USA}

\date{\today}

\begin{abstract}
Discovering low-overhead quantum error-correcting codes is of significant interest for fault-tolerant quantum computation. For hardware capable of long-range connectivity, the bivariate bicycle codes offer significant overhead reduction compared to surface codes with similar performance. In this work, we present ``ZSZ codes'', a simple non-abelian generalization of the bivariate bicycle codes based on the group $\mathbb{Z}_\ell \rtimes \mathbb{Z}_m$. We numerically demonstrate that certain instances of this code family achieve competitive performance with the bivariate bicycle codes under circuit-level depolarizing noise using a belief-propagation and ordered-statistics decoder, with an observed threshold around $0.5\%$.  We also benchmark the performance of this code family under local ``self-correcting'' decoders, where we observe significant improvements over the bivariate bicycle codes, including evidence of a sustainable threshold around $0.095\%$, which is higher than the $0.06\%$ that we estimate for the four-dimensional toric code under the same noise model. These results suggest that ZSZ codes are promising candidates for scalable self-correcting quantum memories. Finally, we describe how ZSZ codes can be realized with neutral atoms trapped in movable tweezer arrays, where a complete round of syndrome extraction can be achieved using simple global motions of the atomic arrays.
\end{abstract}

\maketitle
\tableofcontents

\section{Introduction}

Recent advances in experimental quantum processors have resulted in numerous demonstrations of small-scale quantum error correction (QEC) \cite{Google_SC_1, Google_SC_2, paetznick2024, Hong_2024_H2, Berthusen_2024_4D, reichardt2024tess, daguerre2025, dasu2025, Bluvstein_2023, rodriguez2024magic, reichardt2025neutral, bluvstein2025arch}, a key component of fault-tolerant quantum computation. The sizes of these experiments range from tens to hundreds of physical qubits with logical qubits in the range of a few to a few dozen.

With larger scale quantum processors come the prospects of increasingly sophisticated quantum error-correcting codes. A significant milestone was the discovery of quantum LDPC codes with asymptotically optimal (good) $\llbracket n,k=\mathrm{\Theta}(n),d=\mathrm{\Theta}(n) \rrbracket$ parameters \cite{Panteleev_2022_good, qTanner_codes, DHLV_codes}. Many other exciting advances include the discovery of asymptotically good non-LDPC quantum codes capable of magic state distillation with constant overhead \cite{wills2024constant}, quantum LDPC codes with high rate and transversal non-Clifford gates \cite{lin2024transversal, golowich2024transversal, breuckmann2024cups}, non-LDPC quantum codes with addressable non-Clifford gates \cite{he2025addressable, he2025goodaddressable}, and end-to-end fault tolerance with constant spatial and low temporal overheads with concatenated codes \cite{Yamasaki_2024} and LDPC codes \cite{nguyen2024constant, tamiya2024constant}. There has furthermore been progress in other aspects of fault tolerance such as lowering the overhead for addressable logical Pauli measurements \cite{williamson2024surgery, cowtan2024ssip, swaroop2025adapters, he2025extractors, yoder2025tour}, a sufficient ingredient for addressable logical Clifford gates and the Pauli-based computational model \cite{Litinski_2019}. 

For the near-term, the asymptotically good codes might not be practical due to their large connectivity requirements. Still, one may hope that the mathematical techniques used to discover these codes can give rise to simple small instances that nevertheless have favorable properties.  One such example is arguably a family of quantum LDPC codes called bivariate bicycle (BB) codes, which achieve competitive memory performance with the surface code while occupying considerably fewer physical qubits, at the cost of nonlocal connectivity \cite{BB_codes}. It was later shown that a fault-tolerant architecture consisting of BB codes and surface codes as memory and computational blocks respectively achieved a lower footprint than surface codes on a wide range of practical algorithms \cite{viszlai2024, stein2024, yoder2025tour}. 
However, all BB codes are local (with finite-range interactions) in some finite-dimensional Euclidean space \cite{arnault2025BPT}, which constrains the ultimate scaling of the code parameters with system size  \cite{Bravyi_2009, BPT}. These bounds motivate the study of ways to minimally circumvent them, while hopefully retaining their favorable error-correcting properties.
 
Given that neutral atom platforms have the potential ability to realize the nonlocal interactions required of very high-dimensional codes, this paper asks the question of whether -- with a similar code size and implementation overhead -- there are ``better'' potential codes than BB codes.  In this paper, we will focus on the question of finding quantum memories, including those which may exhibit self-correction and admit passive decoding. Passive QEC offers a few alluring properties over traditional QEC, typically at the expense of a lower threshold. First, it enables a new paradigm of error correction that forgoes mid-circuit measurements and feedback, also known as measurement-free quantum error correction (MFQEC) \cite{Dennis_2002, Ahn_2002, Sarovar_2005}. Second, each decoding ``cycle'' in a passive memory is typically a constant-depth circuit that does not depend on any information from previous cycles, analogous to single-shot error correction \cite{Spielman_1996_red, Bombin_2015}, and so constant-depth logical operations (e.g. transversal) can be performed between each cycle. In other words, a logical cycle of a passive quantum memory takes $O(1)$ time, in units of syndrome extraction cycles. In contrast, the typical logical cycle for surface codes and BB codes requires $\mathrm{\Theta}(d)$ time \cite{Dennis_2002, BB_codes}; logical operations can be performed at a faster rate but at the cost of increased decoding complexity \cite{cain2025corr, zhou2024alg}.
 
In this paper, we introduce the \textbf{ZSZ codes}, a family of quantum LDPC codes where each parity check involves six qubits, and each qubit participates in six parity checks split evenly between three $X$-type and three $Z$-type checks. These codes are a non-abelian generalization of the BB codes where we ``twist'' the associated product involved in the construction. This twist, otherwise known as a semidirect product, involves a relatively simple adjustment to the microscopic rules for how qubits and checks are connected; see Figure \ref{fig:BB vs ZSZ layout} for an illustration. Nonetheless, we demonstrate that this modification can lead to \emph{global} changes in code properties, the most drastic of which is the possibility for passive error correction. Under our experimentally inspired noise model, we observe a threshold around $0.5\%$ for ZSZ codes under $d$ rounds of syndrome extraction and global decoding, which is close to the estimated $0.8\%$ threshold for the surface code under the same noise model. Using a passive ``self-correcting'' decoder, we observe a sustainable threshold around $0.095\%$, which is higher than the estimated $0.06\%$ threshold for the four-dimensional toric code. To the best of our knowledge at the time of writing, $0.095\%$ is the highest observed sustainable threshold for passive decoding of any known quantum LDPC code under similar circuit-level noise. Note that this threshold can potentially increase if measurements are utilized and decoding is performed ``offline'' on a noiseless classical computer; when running the decoder in the measurement-free setting, one needs to also take into account possible faults in its implementation. We finally study the implementation of ZSZ codes as a memory in neutral atom arrays and describe the necessary optical-tweezer movements required to perform syndrome extraction. Similar to the BB codes, we provide a two-dimensional rectangular embedding of ZSZ codes. With respect to this embedding, we then construct a routing protocol for syndrome extraction whose complexity is logarithmic in the horizontal dimension and linear in the vertical dimension. Although our ZSZ routing complexity increases with system size, unlike that of toric and BB codes with constant complexity, the single-shot property of ZSZ codes allows for fewer rounds of syndrome extraction within a logical cycle. A more detailed study of this tradeoff would be important when deciding between ZSZ codes or BB codes for the neutral-atom architecture.

The paper is organized as follows. In Section \ref{sec:two-block codes}, we review the construction of quantum stabilizer codes from two-block matrices and focus our attention on those based on group algebras. In Section \ref{sec:ZSZ codes}, we introduce the ZSZ codes as a special family of these two-block quantum codes and present some geometrical and algebraical arguments relevant to their understanding. In Section \ref{sec:numerics}, we discuss our noise model and present the results of our numerical simulations. In Section \ref{sec:neutral atom implementation}, we discuss the implementation of ZSZ codes in the neutral-atom architecture and present a routing protocol to realize their long-range connectivity. Finally in Section \ref{sec:outlook}, we close with some open questions and concluding remarks. The Appendices contain technical details for the arguments in the main text.

\begin{figure}[t]
    \centering
    \includegraphics[width=\textwidth]{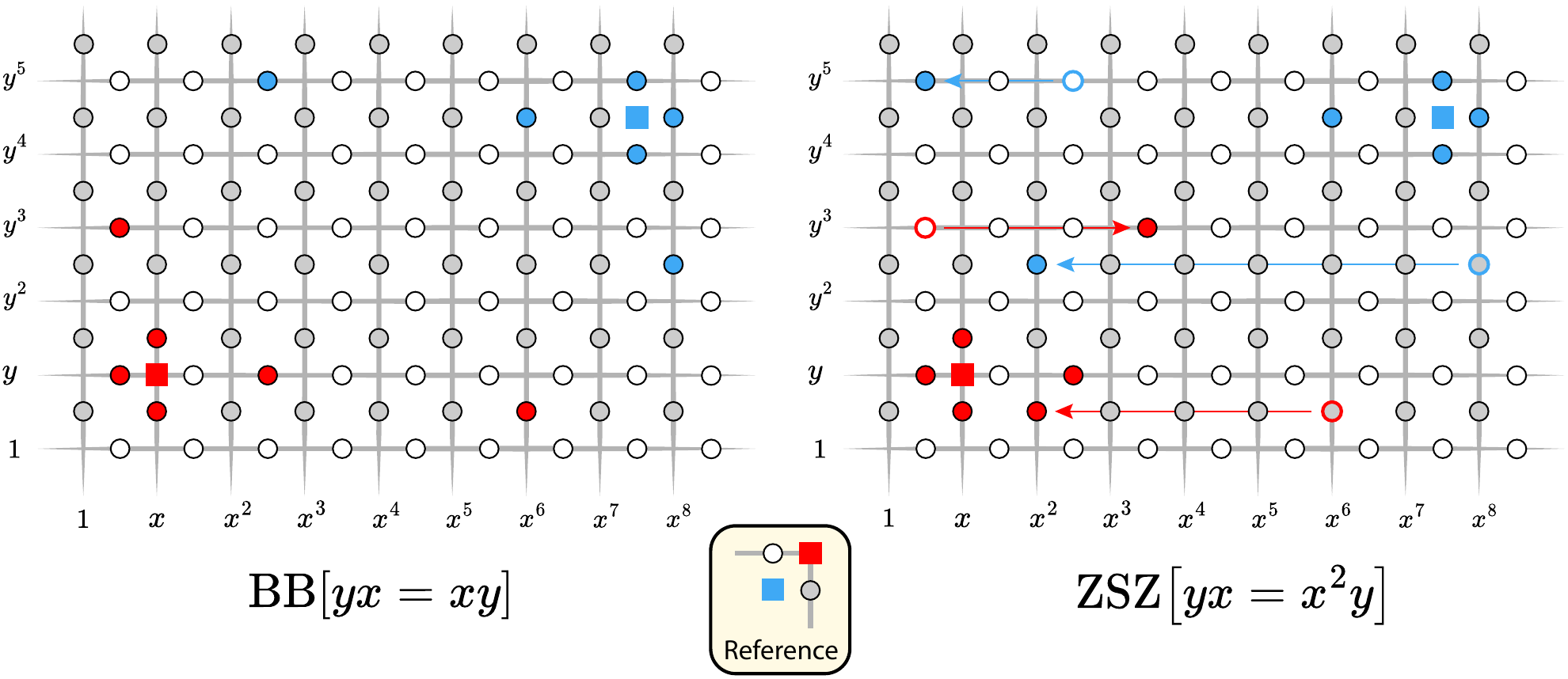}
    \caption{The physical layouts of a BB and a ZSZ code with shared polynomials $a = 1+x^2+y^2$ and $b=1+x^5+y$ are depicted. Qubits (circles) live on the horizontal and vertical links of a $9\times6$ grid with periodic boundaries. The red (blue) square denotes an $X$-check ($Z$-check) and the red-shaded (blue-shaded) circles label its support, which can be computed through $a$ and $b$. The arrows on the right side show how the check support changes upon applying the twist $yx=x^2y$, which converts the BB code to the ZSZ code. In contrast to BB codes, the other checks of ZSZ codes cannot generically be obtained upon translation.}
    \label{fig:BB vs ZSZ layout}
\end{figure}

\subsection{Related works}

Paz-Silva \emph{et al.} \cite{Paz-Silva_2010} design a MFQEC protocol for the 9-qubit Bacon-Shor code and demonstrate a pseudothreshold around $4\times 10^{-5}$ for preparation and gate errors. Heu\ss{}en \emph{et al.} \cite{Heuben_2024} design a flagged-based MFQEC protocol for the 7-qubit Steane code and achieve a pseudothreshold around $6\times 10^{-5}$ with a single-parameter noise model and a two-qubit gate decomposition as well as $6\times 10^{-4}$ with a multi-parameter noise model including native multiqubit gates such as CCZ. Our noise model is relatively close to theirs, but we note that their simulations involve the full state vector and so can account for coherent errors from non-Clifford gates. More recently, Butt \emph{et al.} \cite{butt2025MFQEC} experimentally demonstrate a universal set of measurement-free fault-tolerant gadgets using a combination of the $\llbracket 4,2,2\rrbracket$ ``Iceberg'' code and the $\llbracket 8,3,2\rrbracket$ color code.

The above works focus on specific small-instance codes, and scalable fault tolerance is achieved through concatenation. In contrast, our proposed scheme relies on constant-depth decoding, and scalable fault tolerance is achieved through a growing family of LDPC codes. Recently, Park \emph{et al.} \cite{Park_2025} study this problem for the 2D repetition and 4D toric codes and employ classical reinforcement learning to achieve MFQEC circuits with multiqubit gates that outperform conventional local decoders based on Toom's (sweep) rule in the subthreshold regime. Their tools are general-purpose, and it would be interesting to see if and by how much they can improve passive thresholds when applied to ZSZ codes.


\section{Review of two-block codes} \label{sec:two-block codes}

In this section, we review some of the concepts used in our code construction, starting from an abstract prescription of two-block quantum codes \cite{2block_CSS} and then focusing on specific two-block codes built with the help of a group algebra \cite{wang2023, 2BGA}.

\subsection{Two-block quantum CSS codes}

A Calderbank-Shor-Steane (CSS) code \cite{Calderbank_1996, Steane_1997} is a specific type of quantum stabilizer code \cite{gottesman1997} whose Pauli checks are strictly $X$-type or $Z$-type. The Pauli checks can be neatly packaged into the rows of two binary parity-check matrices $H_X$ and $H_Z$, and the stabilizer commutativity condition becomes the orthogonality condition $H^{}_X H^\transpose_Z = 0$, where addition is performed modulo 2.  
 Such a code is an LDPC (low-density parity-check) code if the rows and columns of $H_X$ and $H_Z$ have an O(1) number of non-zero entries.
 
A two-block CSS code \cite{2block_CSS} builds the above parity-check matrices with the help of two $n\times n$ commuting square matrices $A$ and $B$:
\begin{subequations}\label{eqs:2block H_X, H_Z}
\begin{align}
    H_X &= \left(\, A \;\big|\; B \,\right)  \\
    H_Z &= \left(\, B^\transpose \;\big|\; A^\transpose \,\right)  \, .
\end{align}
\end{subequations}
Since $[A,B]=0$, we have $H^{}_X H^\transpose_Z = AB+BA = [A,B] = 0$,\footnote{$+$ and $-$ are equivalent modulo 2} and thus $H_X$ and $H_Z$ define a valid CSS code. Since $A$ and $B$ are $n\times n$ matrices, there are $2n$ physical qubits and $n$ Pauli checks of each type. Because of the two-block structure above, we also have that
\begin{subequations}
\begin{align}
    \ker{A}\,,\, \ker{B} &\subset \ker{H_X}  \\
    \ker{A^\transpose}\,,\, \ker{B^\transpose} &\subset \ker{H_Z}  \, .
\end{align}
\end{subequations}
where by slight abuse of notation, we write $\ker{A}$ to mean $\ker{A}$ on the left block and zeros on the right block, and vice versa for $\ker{B}$. The task is now to construct the commuting matrices $A$ and $B$, which we can view as two classical linear codes. Ideally, we would like $A$ and $B$ to have large code distances in the hopes that their two-block CSS code will as well. In addition, we would also like $A$ and $B$ to be sparse so that the resulting CSS code is LDPC and hence has inherent fault-tolerant properties \cite{Breuckmann_2021_LDPC}.

\subsection{Two-block group algebra (2BGA) codes}
\label{sec:2BGA}

One method of constructing two sparse, commuting matrices $A$ and $B$ is with the help of a group algebra \cite{wang2023, 2BGA}. Suppose we have a finite group $\mathcal{G}$ of order $n$. A group algebra $K[\mathcal{G}]$, where $K$ is a field, is an object that marries the additive properties of fields and the multiplicative properties of groups. Since we are interested in qubit CSS codes constructed from binary matrices, we will choose $K=\mathbb{F}_2$ as our field. A generic element $a$ of the group algebra $\mathbb{F}_2[\mathcal{G}]$ can be written as a linear combination of group elements with binary coefficients:
\begin{align}\label{eq:group algebra element example}
    a = \sum_{i=1}^{n} c_i g_i \, .
\end{align}
Addition and multiplication are done in a natural way.  For example, $(g_1+g_2) + (g_3+g_4) = g_1+g_2+g_3+g_4$, while $(g_1+g_2)(g_3+g_4) = g_1g_3+g_1g_4+g_2g_3+g_2g_4$ and $g_1 +g_1=0$.

For a 2BGA code, we will choose two elements $a,b \in \mathbb{F}_2[\mathcal{G}]$. Because we are ultimately constructing LDPC codes, we will choose $a$ and $b$ to have a finite number of terms, which we will  see results in sparse parity-check matrices. Our binary matrices $A$ and $B$ are given by a binary matrix representation of our group-algebra elements $a$ and $b$:
\begin{align} \label{eq:binary matrix A and B}
    A = \mathbb{B}[a] \quad,\quad B = \mathbb{B}[b] \, ,
\end{align}
where $\mathbb{B}[a]$ denotes some binary matrix representation of $a \in \mathbb{F}_2[\mathcal{G}]$, which is usually taken as the (binary) regular representation. In the regular representation, a group element maps to a unique basis vector over $\mathbb{F}^n_2$, and its group action becomes an $n \times n$ permutation matrix which encodes the group's multiplication (Cayley) table. For non-abelian groups, we may choose this regular representation to correspond to either left ($L[\cdot]$) or right ($R[\cdot]$) multiplication\footnote{Note that here we are using a slight abuse of notation. Usually the right-regular representation is defined as $R[a] \equiv a^{-1}$ so that the representation condition $R[a]R[b]=R[ab] \equiv b^{-1}a^{-1}$ is satisfied. So when we write $R[a]$ to denote right-multiplication by $a$, we really mean $R[a^{-1}]$ by the standard convention.}, a fact which will prove useful shortly. It follows that $A$ and $B$ are sums of permutation matrices. For example, using the group $\mathcal{G}=\mathbb{Z}_3$ with polynomial $a=1+x$, we get
\begin{align}
    A = \begin{pmatrix}
        1 & 1 & 0 \\
        0 & 1 & 1 \\
        1 & 0 & 1
    \end{pmatrix} \, ,
\end{align}
which is a classical parity-check matrix for the 3-bit repetition code.
A benefit of using the regular representation is that the row and column weights of $A$ and $B$ are immediately bounded by the number of nonzero terms of their corresponding group-algebra elements $a$ and $b$.  Choosing only a small constant number of terms in \eqref{eq:group algebra element example} results in a sparse $A$ and $B$. Commutativity between $A$ and $B$ follows from that of $a$ and $b$, which always holds if $\mathcal{G}$ is abelian. As an example, bivariate bicycle (BB) codes \cite{BB_codes} are 2BGA codes over the abelian group $\mathbb{Z}_\ell \times \mathbb{Z}_m$ for integers $\ell,m>0$.

When $\mathcal{G}$ is non-abelian, $A$ and $B$ are no longer guaranteed to commute since $a$ and $b$ do not necessarily commute anymore. However, recall that for a non-abelian group, we have the choice of choosing either the left-regular or right-regular representations for $A$ and $B$. As such, we choose $A = L[a]$ as the left-regular representation and $B = R[b]$ as the right-regular representation. The associativity of group multiplication then ensures that $a$ and $b$, and thereby $A$ and $B$, commute.

As an example, the symmetric group $\mathrm{S}_3$ on three symbols is the smallest non-abelian group. In cycle notation, all six of its elements are
\begin{align}\label{eq:S_3 group elements}
    \mathrm{S}_3 = \big\{(1),(1,2),(1,3),(2,3),(1,2,3),(1,3,2)\big\} \, .
\end{align}
Since $\abs{\mathrm{S}_3} = 6$, its regular representation will be six-dimensional, and each of the above group actions will map to some $6\times 6$ permutation matrix. Take the transposition $(1,3)$ for instance. Using the ordering of \eqref{eq:S_3 group elements}, it will map to the basis vector $(0,0,1,0,0,0) \in \mathbb{F}^6_2$. Now, acting $(1,3)$ to the left of \eqref{eq:S_3 group elements} gives us $(1,3)\cdot\mathrm{S}_3 = \left\{(1,3),(1,2,3),(1),(1,3,2),(1,2),(2,3)\right\}$, and acting on the right gives us $\mathrm{S}_3\cdot(1,3) = \{(1,3),(1,3,2),(1),(1,2,3),(2,3),(1,2)\}$. Comparing with the ordering of \eqref{eq:S_3 group elements}, the left-regular and right-regular representations of $(1,3)$ are the permutation matrices
\begin{align}
    L[(1,3)] = \begin{pmatrix}\;
        0 & 0 & 1 & 0 & 0 & 0 \\
        0 & 0 & 0 & 0 & 1 & 0 \\
        1 & 0 & 0 & 0 & 0 & 0 \\
        0 & 0 & 0 & 0 & 0 & 1 \\
        0 & 1 & 0 & 0 & 0 & 0 \\
        0 & 0 & 0 & 1 & 0 & 0
    \;\end{pmatrix} \quad,\quad
    R[(1,3)] = \begin{pmatrix}\;
        0 & 0 & 1 & 0 & 0 & 0 \\
        0 & 0 & 0 & 0 & 0 & 1 \\
        1 & 0 & 0 & 0 & 0 & 0 \\
        0 & 0 & 0 & 0 & 1 & 0 \\
        0 & 0 & 0 & 1 & 0 & 0 \\
        0 & 1 & 0 & 0 & 0 & 0 
    \;\end{pmatrix} \, .
\end{align}
One can quickly verify that $L[(1,3)]$ and $R[(1,3)]$ commute, as is guaranteed by the associativity of group multiplication.


\section{ZSZ codes} \label{sec:ZSZ codes}

We are now ready to define our ZSZ codes. Let $\ell,m,q$ be positive integers satisfying \begin{equation}
    q^m \equiv 1 \; \text{(mod $\ell$)}. \label{eq:qm1l}
\end{equation} We define the ZSZ group of order $\ell m$ according to the presentation
\begin{equation}\label{eq:ZSZ presentation}
    \mathbb{Z}_\ell \rtimes_q \mathbb{Z}_m := \langle\, x, y \,|\, x^\ell = y^m = yxy^{-1}x^{-q}=1 \,\rangle \, .
\end{equation}
A generic group element can be written as some word made up of the symbols $x,y$, e.g. $x^2yxy^2$. The relations in \eqref{eq:ZSZ presentation} provide us a way to compute equivalences between different words. For both $\mathbb{Z}_\ell \times \mathbb{Z}_m$ and $\mathbb{Z}_\ell \rtimes_q \mathbb{Z}_m$, we can always reduce all words to a canonical lexicographical form $x^iy^j$, where $i=0,\dots,\ell-1$ and $j=0,\dots,m-1$. The final relation $yxy^{-1}=x^q$ defines the twist that the $y$ terms apply to the $x$ terms upon conjugation and distinguishes between the abelian direct product ($q=1$) and non-abelian semidirect product ($q>1$). In the language of group theory, $\mathbb{Z}_\ell := \langle\, x\,|\,x^\ell=1 \,\rangle$ is a normal subgroup of \eqref{eq:ZSZ presentation}, and $y\in\mathbb{Z}_m$ acts on this subgroup via conjugation, i.e. the automorphism $\varphi_y(x) = x^q$ of $\mathbb{Z}_\ell$, the validity of which is ensured by \eqref{eq:qm1l}.
We define a \textbf{ZSZ code} to be a 2BGA code where the binary matrix $A$($B$) in (\ref{eqs:2block H_X, H_Z}) is the left(right)-regular representation of a $\mathbb{F}_2$-group-algebra element $a$($b$) with respect to a ZSZ group; i.e. $a,b \in \mathbb{F}_2[\mathbb{Z}_\ell \rtimes_q \mathbb{Z}_m]$. We review some basic group theories relevant to this code in Appendix \ref{app:group codes}. For brevity, we denote ZSZ$(\ell,m,q; a,b)$ to be the ZSZ code with group \eqref{eq:ZSZ presentation} and polynomials $a,b \in \mathbb{F}_2[\mathbb{Z}_\ell \rtimes_q \mathbb{Z}_m]$.

A geometrical object called a Cayley graph will be useful to describe how parity checks and qubits are connected in a ZSZ code. For any finite group $G$ equipped with a set of generators $S = \{ s_1,s_2,\dots \}$ that does not include the identity, the (directed) Cayley graph Cay$(G,S)$ is a simple graph where vertices are labeled by group elements, and two vertices share an edge if and only if their corresponding group elements are related by a generator in $S$; e.g. the vertices corresponding to $g_1$ and $g_2$ are connected by an edge if $g_1 = sg_2$ or $g_2=sg_1$ for some $s \in S$. When $G$ is non-abelian, the order of multiplication matters, which motivates the distinction between left and right Cayley graphs. Edges on the left Cayley graph represent left-multiplication by a generator, and edges on the right Cayley graph represent right-multiplication. For example, Cay$(\mathbb{Z}_\ell, \{x\})$ is a ring graph of length $\ell$, and Cay$(\mathbb{Z}_\ell \times \mathbb{Z}_m, \{x,y\})$ is an $\ell\times m$ rectangular lattice with periodic boundaries.

Our lexicographical ordering $x^iy^j$ of the ZSZ group elements motivates a two-dimensional rectangular layout where one axis labels the exponent of $x$ and the other $y$. Because we have a two-block code, we will also have two copies of this group, for a total of $2\ell m$ data qubits, which can be arranged as the links (edges) of a rectangular lattice with periodic boundaries: horizontal and vertical links comprise the qubits in the $A$ and $B$ blocks of \eqref{eqs:2block H_X, H_Z} respectively, see Figure \ref{fig:BB vs ZSZ layout}.

For a BB code, which we remind the reader corresponds to the abelian case $q=1$, it is relatively straightforward to determine the support of the parity checks from the polynomials $a$ and $b$. For an $X$-check at coordinate $(i,j)$ corresponding to the term $x^iy^j$, we examine the monomials in $a$ and $b$ and add their respective $x$ and $y$-exponents to $i$ and $j$. This simple addition is possible because we can commute the $x$ and $y$ terms through each other and combine like terms together. One can also envision these rules as traversing a path on the rectangular lattice: a generic monomial such as $x^\alpha y^\beta$ corresponds to moving $\alpha$ units along the $x$ direction and $\beta$ units along the $y$ direction. Each monomial in $a$ goes to a horizontal qubit, and each monomial in $b$ goes to a vertical qubit. For example, in Figure \ref{fig:BB vs ZSZ layout} with polynomials $a=1+x^2+y^2$ and $b=1+x^5+y$, the $X$-check at position $xy$ connects to horizontal qubits at positions given by $a(xy) = xy+x^3y+xy^3$ and vertical qubits at positions given by $(xy)b = xy+x^6y+xy^2$.

For a ZSZ code, the simple addition rules above do not work because when we move $x$ terms through $y$ terms, the exponents on the $x$ terms will change according to the twist $q$ in \eqref{eq:ZSZ presentation}. Instead, the rules are given by the following ``push-through'' relations:
\begin{subequations}\label{eq:push-through relations}\begin{align}
    (x^i y^j)x^\alpha &= x^{i+q^j\alpha} y^j \, ,  \label{eq:push-through x right} \\
    y^\beta(x^iy^j) &= x^{q^\beta i} y^{j+\beta} \, . \label{eq:push-through y left}
\end{align}\end{subequations}
Note that these push-through relations are inherently encoded in the ZSZ group's left and right Cayley graphs, and so we can recover an analogous geometrical path picture as in the abelian case upon replacing the rectangular lattice with the Cayley graphs of the corresponding ZSZ group. This prescription also encompasses the BB codes since their Cayley graphs are precisely the rectangular lattices used in their analyses. Now, when we see a monomial in $a$ like $x^\alpha y^\beta$, we will read off from right to left\footnote{This particular order follows the multiplication order of multiplying $x^\alpha y^\beta$ on the left: $y^\beta$ is acted first followed by $x^\alpha$.}: we first traverse $\beta$ steps along $y$-edges and then $\alpha$ steps along $x$-edges in the left Cayley graph. When we see a monomial in $b$ like $x^\alpha y^\beta$, we will read off from left to right: we first traverse $\alpha$ steps along $x$-edges and then $\beta$ steps along $y$-edges in the right Cayley graph. See Figure \ref{fig:X-check Cayley graphs} for an explicit example of a ZSZ code with both its underlying left and right Cayley graphs drawn out. In our rectangular layout, we see that the left Cayley graph consists of $m$ copies of $\mathbb{Z}_\ell$ arranged in rows, with neighboring rows related by the group automorphism $\varphi_y(x)=x^q$. The right Cayley graph also has $m$ rows of $\mathbb{Z}_\ell$, but $\varphi_y$ now acts independently within each row rather than between rows: the row associated with $y^j$ transforms according to $\varphi^j_y(x) = x^{qj}$.

The advantage of the nontrivial push-through relations is as follows.   Let $B_r(h)$ be the set of vertices within distance $r$ of vertex $h$ in the Cayley graph of the underlying group.   If $q=1$, i.e. we have a BB code with $p$ distinct monomials $g_1,\ldots, g_p$ in $a$ and $b$, then \begin{equation}
    |B_r(h)| = O(r^{p}). \label{eq:BrscalingBB}
\end{equation}
This bound is very loose, but it is easy to motivate: after $r$ multiplications by $p$ elements, the number of distinct group elements we can reach (up to the left/right separation of the qubits) is $g_1^{r_1}\cdots g_p^{r_p}h$ with $r=r_1+\cdots + r_p$.   The number of choices of $(r_1,\ldots, r_p)$ scales as \eqref{eq:BrscalingBB}.   In contrast, for a ZSZ code, we can have \begin{equation}
    |B_r(h)| = \exp[\mathrm{\Omega}(r)] \, . \label{eq:BrscalingZSZ}
\end{equation}
This can be seen by explicit construction with $q=2$.  Given generators $x$ and $y$ alone and any integer $J = \sum_\nu J_\nu 2^\nu$ where $J_\nu \in \{0,1 \}$ denotes the $\nu$th digit of $J$ in its binary representation, we can express
\begin{equation}
    x^J = x^{J_0}yx^{J_1}\cdots x^{J_{R-1}}yx^{J_R}y^{-R} \, ,
\end{equation}
where $R = O(\log J) \le r/3$ is the number of binary digits that we needed. Since for any given $R$ there are $2^R$ distinct values for $J$, clearly we can reach the number of group elements given by \eqref{eq:BrscalingZSZ} in the desired number of steps.

\eqref{eq:BrscalingZSZ} is desirable because it suggests that the Cayley graph exhibits small-set expansion -- the number of vertices at the boundary of small subsets is proportional to the volume.  Codes whose Cayley graphs are expanding may exhibit linear confinement, which is sufficient to realize single-shot \cite{Quintavalle_2021} and passive error correction \cite{thermal_LDPC,Placke:2024wey}.   Unfortunately, we show in Appendix \ref{app:girth} that ZSZ Cayley graphs have a constant girth, in contrast to Ramanujan graphs \cite{Lubotzky_1988} that have logarithmic girth. However, girth only tells us about the maximal degree of confinement, and linear confinement typically persists well beyond the girth \cite{Sipser_1996, McKenzie_2021}. The parametric improvement of \eqref{eq:BrscalingZSZ} over \eqref{eq:BrscalingBB} suggests that ZSZ codes may have much better performance under single-shot or autonomous decoders, relative to BB codes, a key feature which we will confirm to be the case in extensive simulations in Section \ref{sec:numerics}. For LDPC codes lacking an extensive number of redundant parity checks, like BB and ZSZ codes, we note that all known examples with self-correction exhibit small-set expansion in their Tanner graphs \cite{Bleher_1995, Kenyon_2001, Montanari_2006, thermal_LDPC, gamarnik2024slow,rakovszky2024bottlenecks}.

We probabilistically search through 3-term polynomials in $\mathbb{F}_2[\mathbb{Z}_\ell \rtimes_q \mathbb{Z}_m]$ to construct both $A$ and $B$ according to the 2BGA prescription described in Section \ref{sec:2BGA}. The resulting ZSZ codes will hence have $n=2\ell m$ data qubits with $\ell m$ $X$-checks and $\ell m$ $Z$-checks. Since the total number of CSS parity checks is $n$, the existence of logical qubits will be based on linear dependencies amongst the parity checks; for any CSS code, we generically have $k = n - \mathrm{rank}(H_X) - \mathrm{rank}(H_Z)$. We estimate the minimum distance using the \textsf{QDistRnd} package in GAP \cite{QDistRnd}, which probabilistically searches for low-weight logical operators. The most promising codes that we found from our computer search are listed in Table \ref{tab:ZSZ codes}.   In this numerical search, we did not constrain the power $q$ to be small.  As we will discuss in Section \ref{sec:neutral atom implementation}, this will make it more difficult to realize the optimized code in experiments.

\begin{figure}[t]
    \centering
    \includegraphics[width=\textwidth]{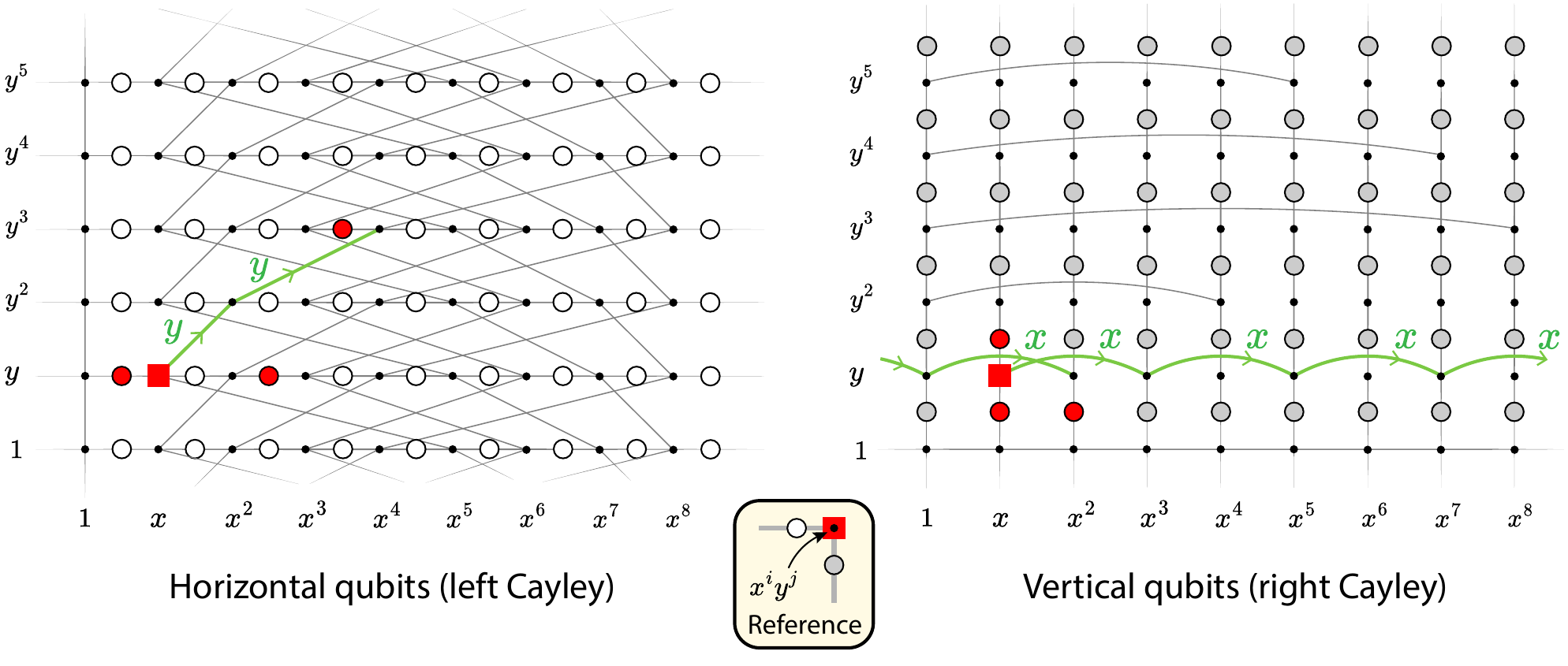}
    \caption{The underlying structure of ZSZ$(9,6,2;a,b)$ with polynomials $a=1+x^2+y^2$ and $b=1+x^5+y$ is depicted. Horizontal qubits (white circles) are indexed by the coordinates of the black dots to their right, while vertical qubits (gray circles) are indexed by the coordinates of the black dots above them. \textbf{Left:} The horizontal support (red circles) of an $X$-check (red square) at location $xy$ is determined by walks on the left Cayley graph (gray lines) according to the monomials in $a$; the path associated with $y^2$ is explicitly drawn in green. \textbf{Right:} The vertical support follows similarly but according to the right Cayley graph and the monomials in $b$. A single horizontal edge is drawn within each row, and the other edges are given by its translations. For $Z$-checks (not shown), the roles of horizontal and vertical qubits are exchanged.}
    \label{fig:X-check Cayley graphs}
\end{figure}

\begin{table}[t]
\centering\renewcommand{\arraystretch}{1.5}
\begin{tabular}{@{}c|c|c|c|c|c@{}}
\hline
Decoding & Name & $\llbracket n,k,d \rrbracket$ & $\ell,m,q$  & $A$ & $B$  \\
\hline
\multirow{8}{*}{$d$ rounds}
    & ZSZ80 & $\llbracket 80,2,\leq10 \rrbracket$ & 5,8,2 & $1+x^4y^4+x^4y$ & $1+x^3+x^2y^7$ \\
    & ZSZ108 & $\llbracket 108,2,\leq12 \rrbracket$ & 3,18,2 & $1+x+y^3$ & $1+xy+x^2y^{12}$ \\
    & ZSZ160 & $\llbracket 160,2,\leq16 \rrbracket$ & 5,16,2 & $1+y^9+x^4y^{14}$ & $1+xy^{13}+x^2y^{13}$ \\
    & ZSZ180 & $\llbracket 180,2,\leq18 \rrbracket$ & 3,30,2 & $1+x^2+y^3$ & $1+xy^{18}+xy^{19}$ \\
    &  ZSZ162  & $\llbracket 162,8,\leq 10 \rrbracket$                             &   27,3,10  & $x^8+x^{18}y+x^{25}y^2$                    & $x^{21}+x^{14}y+x^{16}y^2$ \\  
    & \; ZSZ288-1 \; & $\llbracket 288,12,\leq 16 \rrbracket$                             &  24,6,5    & $xy+x^3y^2+x^5y^3$                    & \;$x^{10}y^3+x^{23}y^4+x^{21}y^4$\;  \\
    & ZSZ360-1 & $\llbracket 360,16,\leq20 \rrbracket$ & 30,6,19 & $1+x^8y^2+x^{26}y^4$ & $1+x^{21}y^4+x^5y^5$ \\
    & ZSZ360-2 & \;$\llbracket 360,20,\leq20 \rrbracket$\; & 30,6,19 & $1+x^9y+x^{26}y^4$ & $1+x^4y^5+x^{16}y^3$ \\               
\hline
\multirow{5}{*}{passive}  
& ZSZ144-3 & $\llbracket 144,12,\leq8 \rrbracket$ & 12,6,5 & $y+x^{2}y^2+x^8y^{3}$ & $x^4+x^6y^4+xy^5$ \\
& ZSZ288-2 & $\llbracket 288,12,\leq 8 \rrbracket$ & 24,6,5 & $1+x^{14}+x^3y^{3}$ & $1+y+x^{8}y^5$ \\
& ZSZ360-3 & $\llbracket 360,12,\leq20 \rrbracket$ & 30,6,11 & \;$x^{12}y+x^{21}y^4+x^4y^{5}$\; & \;$x^4y^3+x^{20}y^3+x^{22}y^4$\; \\
& ZSZ540 & $\llbracket 540,16,\leq 12 \rrbracket$ & 45,6,19 & $1+x^{34}+x^4y^{1}$ & $1+x^{38}+x^{2}y^5$ \\
& ZSZ756 & $\llbracket 756,24,\leq 20 \rrbracket$ & \; 42,9,25 \; & $1+x^{4}y^4+x^6y^{3}$ & $1+x^{36}y^5+x^{21}y^7$ \\
\hline
\end{tabular}
\caption{Candidate ZSZ codes and their parameters and displayed, sorted accordingly to numerical benchmarking against various error models and decoders. Code distances are numerically estimated using the GAP package \textsf{QDistRnd} \cite{QDistRnd}.}
\label{tab:ZSZ codes}
\end{table}


\section{Numerical simulations}\label{sec:numerics}

In this section, we present several instances of ZSZ codes and numerically simulate quantum memory experiments under circuit-level depolarizing noise with three types of decoders. For all the numerical simulations, we execute the following protocol:
\begin{enumerate}
    \item Noiseless preparation of data qubits in $\ket{0}^{\otimes n}$
    \item Multiple rounds of single-ancilla syndrome extraction, alternating between $X$-type and $Z$-type checks
    \item Transversal measurement of all data qubits in the $Z$ basis
\end{enumerate}
For the purposes of our memory simulations, step 1 can also be interpreted as starting in the logical $\ket{\overline{0}}^{\otimes k}$ state. The number of syndrome extraction cycles in step 2 is variable and will depend on the type of memory experiment that we are trying to simulate. For each syndrome extraction cycle, we sequentially extract first the $X$-syndrome and then the $Z$-syndrome; specially scheduled syndrome extraction circuits may further reduce logical error rates than what we report. For step 3, a final noiseless $Z$-syndrome can be inferred from the data qubit measurement outcomes that will be used to return the system exactly to the codespace. Logical $\overline{Z}$ measurement outcomes on all logical qubits can then be read off from the corrected data qubit measurement outcomes. We declare success if all logical $\overline{Z}$ measurement outcomes are +1 and failure otherwise.

\subsection{Noise models and decoders}

Our circuit-level noise model is parameterized by a single physical noise strength $p$. For simulation efficiency, we restrict to local depolarizing noise, which is a standard benchmark for preliminary analyses of new codes. Any qubit undergoing a single-qubit gate, including idling, experiences single-qubit depolarizing noise with probability $p/10$: one of the three Paulis is applied randomly. Two qubits undergoing a two-qubit gate experience two-qubit depolarization with probability $p$: one of the fifteen two-qubit Paulis is applied randomly. Ancilla qubits are incorrectly measured and incorrectly reset with probability $p$. This noise model closely resembles the actual physical noise observed in recent hardware experiments involving trapped ions \cite{Quantinuum_H2_2024} and neutral atoms \cite{Evered_2023}. We use the Python package \textsf{Stim} \cite{gidney2021stim} to perform all circuit-level noise simulations in this work.

The first simulation that we consider involves $d$ syndrome extraction cycles, followed by global decoding of all $d+1$ measured $Z$-syndromes, including the final noiseless syndrome. We employ a belief-propagation and ordered-statistics decoder (BP+OSD) \cite{Roffe_LDPC_Python_tools_2022}, which performs local message-passing on a decoding (hyper)graph before inferring a correction \cite{Panteleev_2021_OSD}. BP+OSD has been previously demonstrated to achieve good performance across a breadth of qLDPC codes \cite{roffe_decoding_2020, Higgott_2023, BB_codes}, and as such is widely expected to be a general-purpose decoder for all qLDPC codes. We construct our decoding graph by laying out $d+1$ copies of the code's $Z$-Tanner graph and inserting ``detector'' qubits between equivalent syndrome nodes, which can also be interpreted as taking a homological product with a length-($d+1$) repetition code. For simulation efficiency, we use this ``phenomenological'' decoding graph rather than the larger circuit-level detector graph outputted by \textsf{Stim}, in essence trading some accuracy for speed. For decoding, we configure BP+OSD with 1000 maximum iterations of ``min-sum'' BP followed by ``combination-sweep'' OSD with search depth 5.

The second simulation that we consider involves $\geq 100$ syndrome extraction cycles, using a single-shot, local ``greedy'' decoder inspired by Glauber dynamics for self-correcting memories. In the classical coding literature, this greedy decoder is more commonly known as the ``flip'' decoder \cite{Gallager_1962, Sipser_1996}: local $X$-corrections are greedily applied one qubit at a time to lower the $Z$-syndrome weight. From the physics perspective, the greedy decoder can be viewed as a zero-temperature Gibbs sampler that locally tries to minimize the energy (syndrome weight). Noise that corrupts the output of the greedy decoder, such as an incorrect input due to syndrome measurement errors or the decoding algorithm being imperfect itself, can all be accounted for by ``raising the temperature''; see Appendix \ref{app:glauber} for details. For certain classical and quantum codes, one can leverage the slow mixing time of low-temperature Gibbs sampling \cite{levin2008markov, gamarnik2024slow,rakovszky2024bottlenecks} to bound the performance of the greedy decoder \cite{Thomas1989, Dennis_2002, Alicki_2010, thermal_LDPC}. In addition, if the code is LDPC, one can ``sweep'' through all data qubits in constant time. After every syndrome extraction cycle, we sweep through all data qubits once and apply the greedy decoder to the $Z$-syndromes for each qubit, corrupting its output with probability $p$: if the decoder outputs an $X$-correction, we do not apply it with probability $p$ and vice versa. This last error is to approximate the situation where we perform ``passive'' (measurement-free) error correction by implementing one sweep of the greedy decoder as a constant-depth noisy circuit within the quantum computer itself\footnote{The greedy decoder involves a majority function that cannot be implemented by a Clifford circuit and so falls outside of the simulation regime of \textsf{Stim}. Thus, we stick to this simpler ``phenomenological'' error model for the decoding step.}. Details regarding passive error correction and its implementation can be found in Appendix \ref{app:passive decoder}. After the final transversal $Z$-measurement, we decode the final syndrome with belief-propagation and localized statistics decoding (BP+LSD), a variant of BP+OSD that reduces the runtime of OSD for large LDPC codes by leveraging the percolation structure of local errors \cite{hillmann2024lsd}. We configure BP+LSD with 1000 maximum iterations of ``min-sum'' BP followed by order-5 ``combination sweep'' LSD. 

Both simulations do not make use of the possibility for erasure checks --- the position-resolved detection of qubit leakage due to physical errors --- as has been demonstrated in super-conducting qubits and neutral-atom arrays~\cite{wu2022erasure,teoh2023dual, kubica2023erasure, ma2023high, scholl2023erasure, levine2024demonstrating, senoo2025high, zhang2025leveraging, bluvstein2025architectural}. Numerical studies have affirmed that the information garnered from erasures can significantly improve thresholds, particularly when detected mid-circuit and even as a ``delayed'' erasure~\cite{grassl1997codes, wu2022erasure, sahay2023high, baranes2025leveraging}. While the ZSZ code already localizes its information to a degree to enable passive decoding, we expect that further gains in performance should be achievable with the use of erasures. We leave these numerical investigations to future work. 

\subsection{Performance}

The simulation results for $d$ rounds of syndrome extractions and global decoding are displayed in Figure \ref{fig:d rounds BPOSD}. We plot both the logical block error rate (BLER), in other words the probability that any logical qubit is measured incorrectly, as well as the logical error rate for any single logical qubit. For a code block encoding $k$ logical qubits with BLER $\overline{p}$, we estimate the single-logical-qubit error rate as $\overline{p}_1 = 1 - (1-\overline{p})^{1/k}$. We also simulate the performance of the $\llbracket 144,12,12 \rrbracket$ (BB144) and $\llbracket 288,12,18 \rrbracket$ (BB288) BB codes from \cite{BB_codes} under the same BP+OSD decoder as well as several (unrotated) 2D surface codes decoded with minimum-weight perfect-matching (MWPM), using the \textsf{PyMatching} package in Python.

\begin{figure}[t]
\centering
\includegraphics[width=0.48\textwidth]{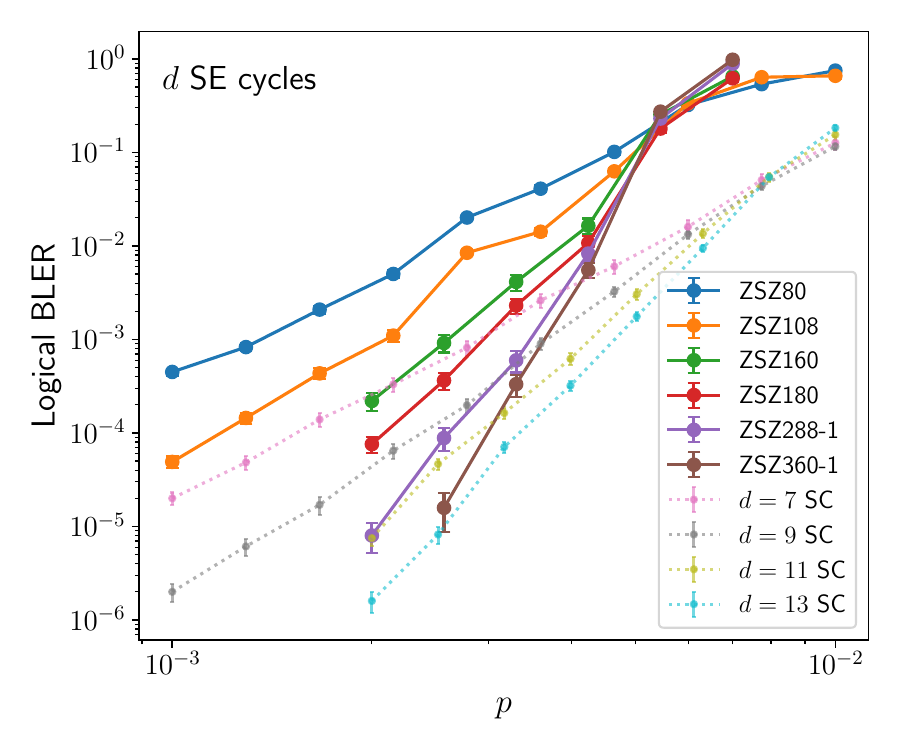}\hfill
\includegraphics[width=0.48\textwidth]{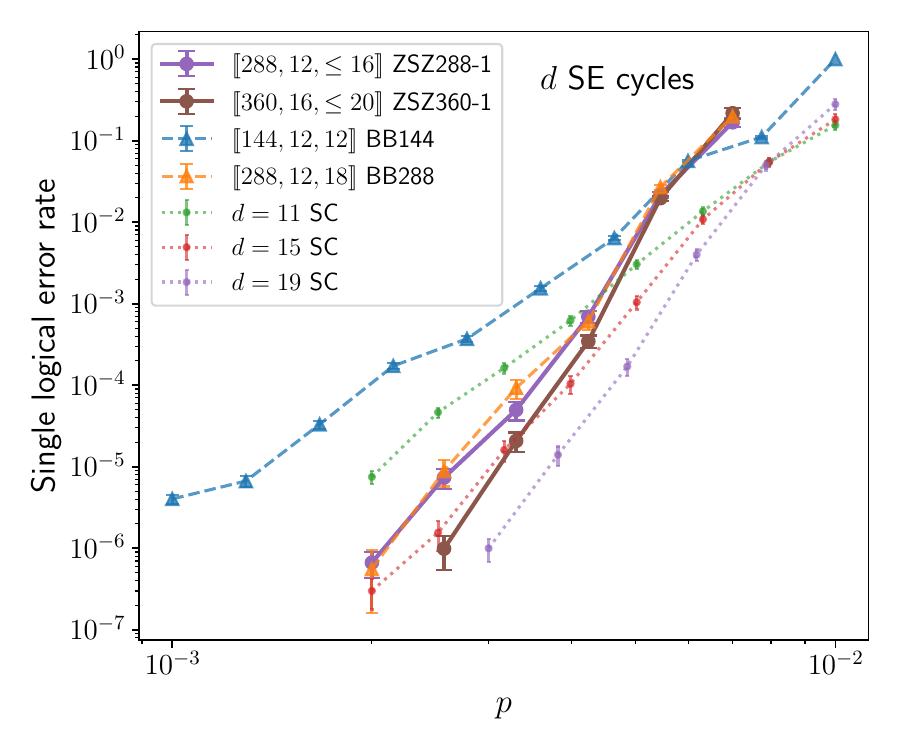}
\caption{Numerical simulation results are displayed for global decoding over $d$ syndrome extraction (SE) cycles. \textbf{Left:} The logical block error rate (BLER) is plotted as a function of the physical noise strength $p$ for a handful of ZSZ codes in Table \ref{tab:ZSZ codes} under ``$d$ rounds'' decoding. For comparison, we also plot the performance of the (unrotated) surface code (SC) (dotted curves). \textbf{Right:} The single logical error rate is plotted as a function of $p$ for two high-rate ZSZ codes. For comparison, we also plot the performance of the surface code (dotted) as well as two BB codes (dashed) from \cite{BB_codes}. Uncertainties are given by the standard error.}
\label{fig:d rounds BPOSD}
\end{figure}

We observe an intersection of the ZSZ BLER curves around $p \approx 0.5\%$, below which the BLER falls exponentially with the code size. This exponential decay is indicative of subthreshold behavior where we would generically expect $\overline{p} \sim p^{-\mathrm{\Theta}(d)}$. Thus, we optimistically declare $p_\mathrm{th} \approx 0.5\%$ as the threshold of ZSZ codes under our choice of decoder and circuit-level noise model. For comparison, the threshold of the 2D surface code with MWPM decoding under the same noise model is observed to be approximately $0.8\%$\footnote{We note that these thresholds can be slightly improved by using depth-optimized syndrome extraction circuits (we sequentially extract the $X$-syndrome and then the $Z$-syndrome).}. Interestingly enough, we also see that some higher-distance ZSZ codes, such as ZSZ180 with estimated code distance 18, have a worse BLER performance than lower-distance ZSZ codes, such as ZSZ288-1 with estimated code distance 16. We attribute this discrepancy to ancillary hook errors during syndrome extraction: a single fault on the ancilla can propagate to multiple data qubits and lower the total \emph{fault distance} of the code. For LDPC codes, this propagation is a constant and so is unimportant from an asymptotic standpoint, but for small-to-moderate length codes this constant can matter. Our ZSZ codes have check weight 6, and so the weight of any hook error is at most 3; see Appendix \ref{app:syndrome extraction} for more details. Note that (unrotated) surface codes do not suffer from this issue \cite{Fowler_2012}. When we normalize by the number of logical qubits and plot the single logical error rate, we notice that the gap between the LER curves of the high-rate ZSZ codes (ZSZ288-1 and ZSZ360-1) and the surface codes closes. In particular, the subthreshold LER slope of ZSZ288-1 closely follows that of BB288 as well as the distance-15 surface code, and the LER slope of ZSZ360-1 closely follows that of the distance-19 surface code. This last observation suggests that ZSZ360-1, under our naive syndrome extraction schedule, achieves comparable performance to a surface code of similar length and distance while encoding 16$\times$ more logical qubits.

\begin{figure}[t]
\centering
\includegraphics[width=0.48\textwidth]{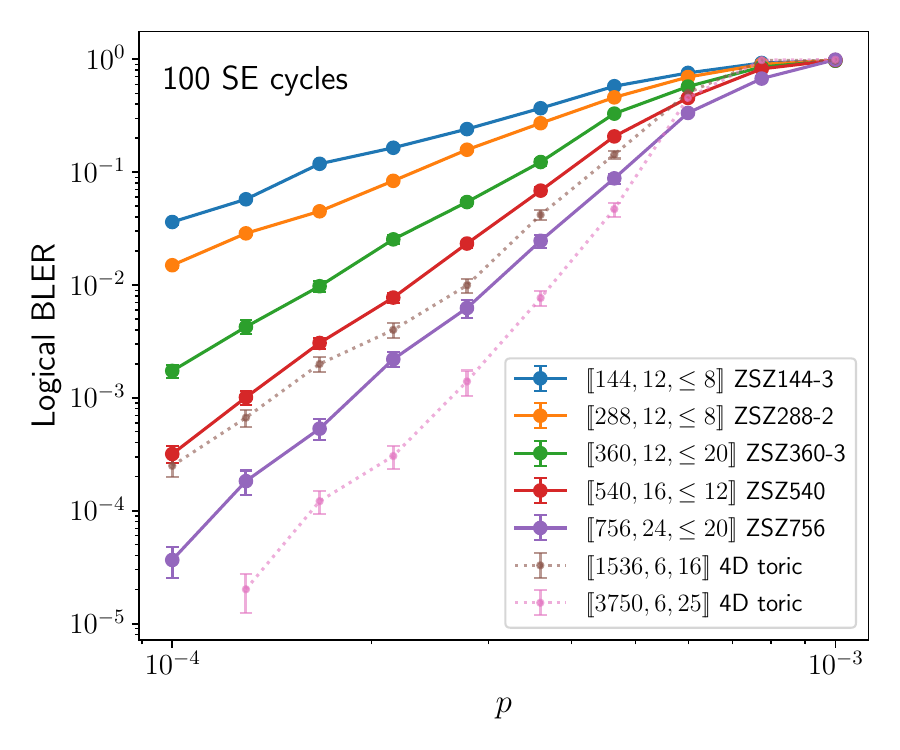}\hfill
\includegraphics[width=0.48\textwidth]{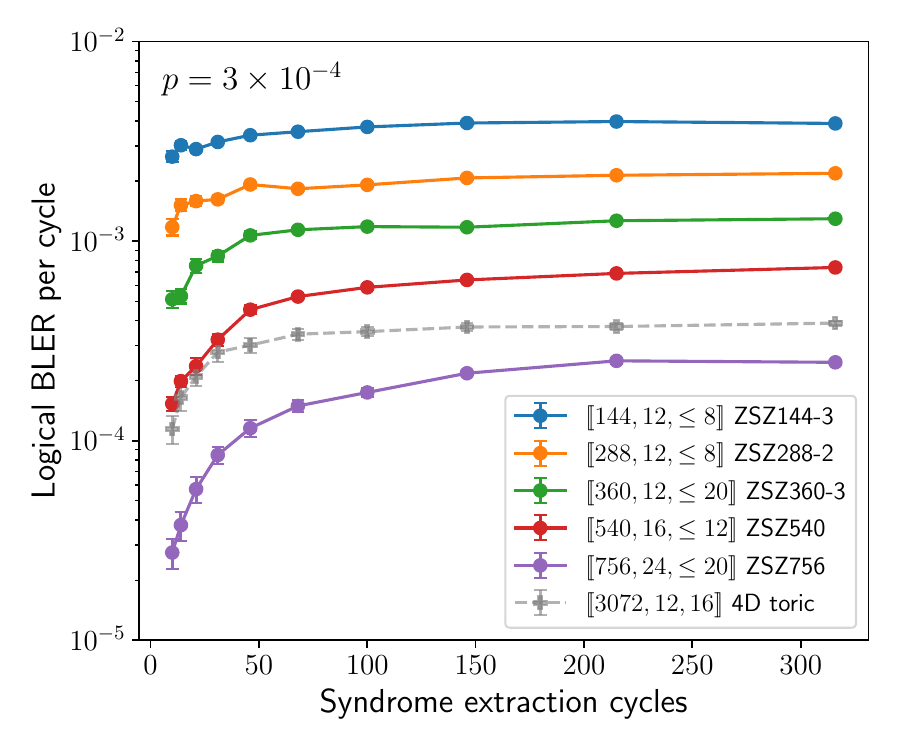}
\caption{Numerical simulation results are displayed for noisy local (passive) decoding over many syndrome extraction (SE) cycles. \textbf{Left:} The logical block error rate (BLER) over 100 SE cycles is plotted as a function of the physical noise strength $p$ for the ZSZ codes in Table \ref{tab:ZSZ codes} under passive decoding. \textbf{Right:} The logical BLER per SE cycle with fixed $p=3\times 10^{-4}$ is plotted as a function of the total number of SE cycles. For comparison, we also plot the performance of the 4D toric code with linear sizes $L=4,5$. Uncertainties are given by the standard error.}
\label{fig:passive decoding}
\end{figure}

The simulation results for $\geq 100$ rounds of syndrome extraction and passive decoding are displayed in Figure \ref{fig:passive decoding}. Similar to the previous situation with $d$ rounds of decoding, we observe subthreshold behavior below a noise strength $p\approx0.1\%$ when decoding a cumulative of 100 syndrome extraction cycles. A more detailed numerical analysis near the observed crossing point can be found in Appendix \ref{app:passive threshold}. In order to rule out finite-size effects, we also fix the noise strength $p$ and vary the number of syndrome extraction cycles. We see that the normalized logical error rates \emph{per cycle} for our ZSZ codes stabilize at $\gtrsim 100$ syndrome extraction cycles, suggesting that our numerics at 100 cycles is a good indicator of the asymptotic behavior at long times. As such, we optimistically declare $p_\mathrm{th} \approx 0.1\%$ as the \emph{sustainable} threshold of ZSZ codes under our model of passive decoding. For comparison, we perform the same simulations for the 4D toric code and estimate its sustainable threshold under passive decoding to be approximately $0.06\%$; see Appendix \ref{app:passive threshold} for details. In particular, we see that ZSZ756 outperforms four copies of the $L=4$ 4D toric code ($\llbracket 6144,24,16 \rrbracket$) while using only $\approx 12\%$ of the number of data qubits for the same number of logical qubits. We also simulate several instances of BB codes and do not observe any evidence of a sustainable threshold for passive decoding, suggesting that ZSZ codes and BB codes can have drastically different qualitative behaviors at moderate code lengths; see Appendix \ref{app:passive BB} for details on the BB code simulations.


\section{Implementation in neutral atom arrays} \label{sec:neutral atom implementation}

In this section we discuss the implementation of ZSZ codes on the neutral atom platform. In this platform, each physical qubit is encoded in two long-lived states of a neutral atom \cite{Saffman2010, Kaufman_2021, Bluvstein_2021, Cong_2022, Bluvstein_2022, Jenkins_2022}. The atoms are optically trapped by a Spatial Light Modulator (SLM), and their positions can be rearranged using Acousto-Optic Deflector (AOD) optical tweezers \cite{Beugnon_2007, Bluvstein_2022, bluvstein2024logical}. Typically, an SLM is used to generate a static two-dimensional array of traps, and AODs are used to transfer atoms between different SLM traps to facilitate long-range gates~\cite{bluvstein2024logical}. A typical AOD reconfiguration step, or ``grid transfer'', consists of picking up a subgrid of qubits from static SLM traps onto dynamical AOD traps, moving the AOD traps to unoccupied SLM traps and then dropping off the atoms onto these SLM traps. Importantly, parallel AOD rows and columns should not intersect to minimize disturbance of static qubits; in other words, the ordering of atoms participating in an AOD move should remain unchanged when moves are constrained to a line. For example, Figure \ref{fig:1D riffle shuffle} illustrates how a ``riffle shuffle'' permutation of atoms arranged in a 1D line can be executed with a constant number of AOD moves and auxiliary SLM traps. This riffle shuffle is a key subroutine of the 1D permutation algorithm by Xu \emph{et al.} \cite{Xu_2024_constant}, which can implement any 1D permutation of $N$ atoms using $\log_2 N$ riffle shuffles and $\lfloor N/2 \rfloor$ additional SLMs for scratch space.

\begin{figure}[t]
    \centering
    \includegraphics[width=0.6\textwidth]{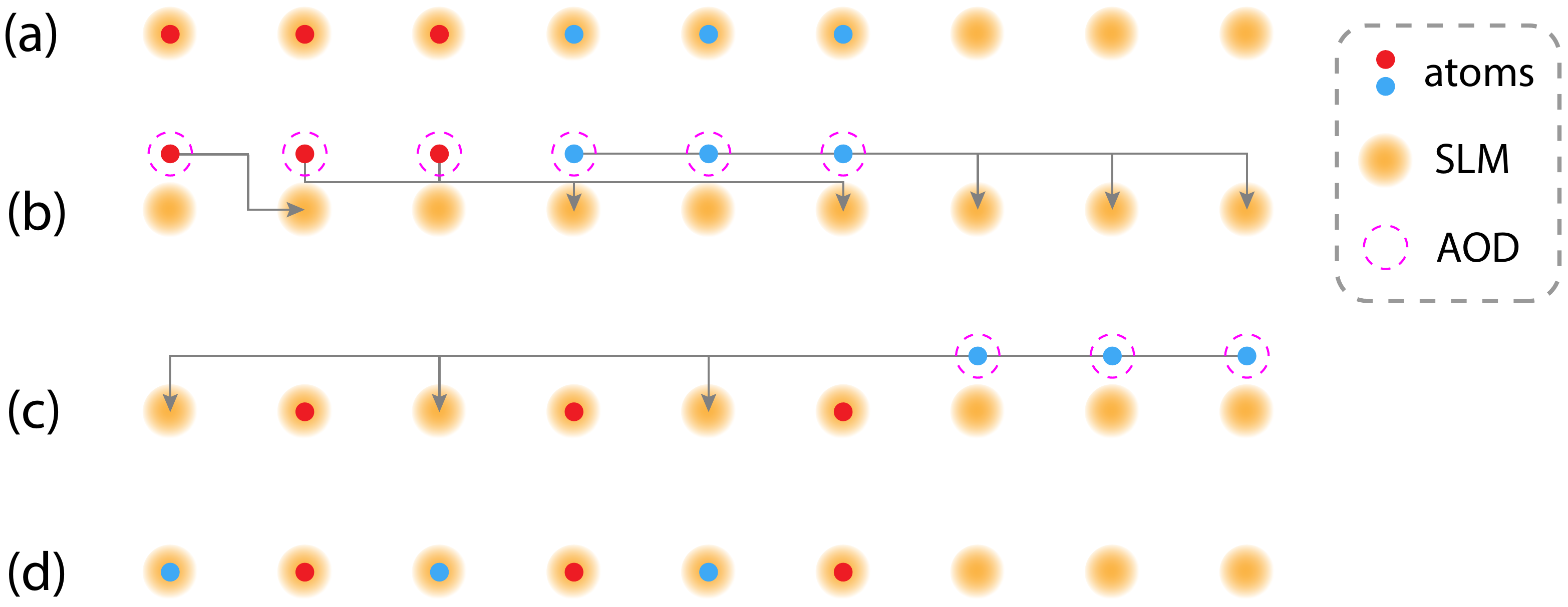}
    \caption{A step-by-step visualization of a 1D riffle shuffle of $N$ atoms (6 shown) using $\lfloor N/2 \rfloor$ auxiliary SLM traps. \textbf{(a)} The initial configuration of atoms. \textbf{(b)} A subset of atoms (blue) is transferred by AODs to the auxiliary SLM traps, and the remaining atoms (red) are moved to some order-preserving subset of the original $N$ sites, in this case the even sites. \textbf{(c)} The blue atoms are then transferred back to the remaining unoccupied SLM traps of the initial $N$ sites. \textbf{(d)} The final configuration of the atoms is displayed.}
    \label{fig:1D riffle shuffle}
\end{figure}

Recall that the semidirect product structure of a ZSZ code naturally gives rise to an $\ell \times m$ rectangular embedding (Figure \ref{fig:X-check Cayley graphs}) where the horizontal axis labels the group elements of the normal subgroup $\mathbb{Z}_\ell$ and the vertical axis those of $\mathbb{Z}_m$. We propose a single-ancilla syndrome extraction strategy, where we initialize a syndrome qubit for every parity check and couple the syndrome qubit to its corresponding data qubits; see Appendix \ref{app:syndrome extraction} for further details. After data-syndrome coupling, we can then either measure the syndrome qubits to extract the classical syndrome bitstring or further process the syndrome qubits for passive decoding. We will now show that a complete round of $X$-syndrome extraction can be performed using $O(\log \ell)$ AOD moves for the horizontal sector and $O(m)$ moves for the vertical sector. The $Z$-syndrome extraction proceeds analogously but using the transpose matrices (i.e. negating the exponents) as well as interchanging the roles of the left and right Cayley graphs, with the same routing complexity.

To extract the $X$-syndrome, we proceed in two steps. First, we couple the syndrome and horizontal data qubits given by the polynomial $a$, which corresponds to routing on the ZSZ group's left Cayley graph. For each monomial $x^\alpha y^\beta$ in $a$, we can implement the $y^\beta$ part using a combination of a vertical cyclic shift of all rows and a 1D horizontal ``grid-type'' permutation of all columns. The number of AOD moves required for this part is generically $O(\log\ell)$, stemming from the 1D permutation complexity. The $x^\alpha$ part can be implemented with a simple horizontal cyclic shift; see Algorithm \ref{alg:ZSZ left routing} for the explicit steps and Figure \ref{fig:implementation}a for a visualization. Hence, the total routing complexity to couple the horizontal data qubits is $O(\log\ell)$.
Second, we couple the syndrome and vertical data qubits given by the polynomial $b$. For each monomial $x^\alpha y^\beta$ in $b$, we can implement the $x^\alpha$ part using horizontal cyclic shifts among the rows; each row will correspond to a distinct cyclic shift given by \eqref{eq:push-through x right}. Using grid transfers and $n$ additional SLM traps as scratch space on the right side, we can perform these cyclic shifts as follows. Using AODs, pick up the entire grid of $n$ syndrome qubits and translate it horizontally, dropping off each row into SLM traps according to their respective shift amount. Then, pick up all the atoms in the scratch space and merge it on the left side of the remaining atoms using additional grid transfers. The number of AOD moves required for this part is $O(m)$, since the rows needs to be sequentially dropped off.
Finally, we can implement the $y^\beta$ part of $b$ with a global vertical cyclic shift. See Algorithm \ref{alg:ZSZ right routing} for the explicit steps and Figure \ref{fig:implementation}b for a visualization.

\begin{algorithm}
    \DontPrintSemicolon
    \caption{ZSZ left-action routing with grid transfers}\label{alg:ZSZ left routing}
    \Input{ZSZ parameters $(\ell,m,q)$ and a polynomial $x^\alpha y^\beta$}
    \Return{A permutation of an $\ell \times m$ rectangular lattice of sites}
    \tcp{$y^\beta$ part, requiring $\lfloor m/2 \rfloor$ additional rows above and $\lfloor \ell/2 \rfloor$ columns on the right for scratch space}
    $\sigma = [q^\beta i \;\%\;\ell \;\textbf{ for }\; i=0 \textbf{ to } \ell-1]$ \;
    Apply a suitable 1D permutation algorithm (e.g. Algorithm 1 of \cite{Xu_2024_constant}) to implement $\sigma$ on all $\ell$ columns, ordered from left to right. \;
    $\beta' = \min(\beta,m-\beta)$ \;
    \uIf{$\beta'==\beta$}{
        Shift the entire lattice upwards by $\beta'$ units. \;
        Pick up the top $\beta'$ rows of atoms and shift them $m$ units downwards to the bottom.
    }
    \Else{
        Pick up the bottom $\beta'$ rows of atoms and shift them $m$ units upwards to the top. \;
        Shift the entire lattice downwards by $\beta'$ units.
    }
    \tcp{$x^\alpha$ part, requiring $\lfloor \ell/2 \rfloor$ additional columns for scratch space}
    $\alpha' = \min(\alpha,\ell-\alpha)$ \;
    \uIf{$\alpha'==\alpha$}{
        Shift the entire lattice to the right by $\alpha'$ units. \;
        Pick up the rightmost $\alpha'$ columns of atoms and shift them $\ell$ units to the left side.
    }
    \Else{
        Pick up the leftmost $\alpha'$ columns of atoms and shift them $\ell$ units to the right side. \;
        Shift the entire lattice to the left by $\alpha'$ units.
    }
\end{algorithm}

\begin{algorithm}
    \DontPrintSemicolon
    \caption{ZSZ right-action routing with grid transfers}\label{alg:ZSZ right routing}
    \Input{ZSZ parameters $(\ell,m,q)$ and a polynomial $x^\alpha y^\beta$}
    \Return{A permutation of an $\ell \times m$ rectangular lattice of sites}
    \tcp{$x^\alpha$ part, requiring $\ell$ additional columns on the right for scratch space}
    $S = [q^j\alpha \;\%\;\ell \;\textbf{ for }\; j=0 \textbf{ to } m-1]$ \;
    $S', I = \mathrm{Sort}(S)$ \quad \tcp{Sorted with index set $I$}
    $\mathrm{\Lambda} \leftarrow$ the entire lattice \;
    \For{$i=1$ \KwTo $m-1$}{
        Shift $\mathrm{\Lambda}$ to the right by $S'[i]$ units. \;
        Drop off row $I[i]$. \;
        $\mathrm{\Lambda} \leftarrow$ $\mathrm{\Lambda} - \mathrm{row}\, I[i]$
    }
    Pick up the $\ell$ columns of scratch space and shift them $\ell$ units to the left side. \;
    \tcp{$y^\beta$ part, requiring $\lfloor m/2 \rfloor$ additional rows above for scratch space.}
    $\beta' = \min(\beta,m-\beta)$ \;
    \uIf{$\beta'==\beta$}{
        Shift the entire lattice upwards by $\beta'$ units. \;
        Pick up the top $\beta'$ rows of atoms and shift them $m$ units downwards to the bottom.
    }
    \Else{
        Pick up the bottom $\beta'$ rows of atoms and shift them $m$ units upwards to the top. \;
        Shift the entire lattice downwards by $\beta'$ units.
    }
\end{algorithm}

To summarize, for syndrome extraction of a ZSZ code with $n=2\ell m$ data qubits, we require $O(\log\ell)$ grid transfers for routing on the left Cayley graph and $O(m)$ grid transfers on the right Cayley graph.
For the ZSZ codes with logarithmic diameter, their dimensions satisfy $m=O(\log\ell)$ which gives a routing complexity of $O(\log\ell)$ for the vertical sector, the same as that of the horizontal sector. Note that all the candidate ZSZ codes for passive decoding in Table \ref{tab:ZSZ codes} have small $m$ relative to $\ell$. 
For comparison, a complete round of syndrome extraction has an (AOD) routing complexity of $O(1)$  for BB codes since all the required couplings can be realized in parallel with cyclic shifts \cite{viszlai2024}, similar to the case with the surface code. For generic hypergraph product codes, a complete round of syndrome extraction has $O(\log n)$ routing complexity \cite{Xu_2024_constant}, coming from the 1D permutation complexity of rearranging rows and columns. So for ZSZ codes with $m = O(\log\ell)$, their routing complexity for syndrome extraction is the same as that for hypergraph product codes, but not as fast as that of BB and surface codes.

Constantinides et al. \cite{constantinides2024routing} proposed a hardware upgrade that allows for ``selective transfers'': within a subgrid formed by AODs, we can select an arbitrary subset of atoms to be left behind during the grid transfer. The process involves deepening the SLM trap potential on the selected sites so that the atoms are not picked up by the AOD tweezers. They show that, using these selective transfers, arbitrary permutations of the sites on a 2D lattice have routing complexity $O(\log n)$. Since their result encompasses arbitrary permutations of the atoms, it applies to the ZSZ codes, and so the routing complexity of syndrome extraction with these selective transfers becomes $O(\log n)$ irrespective of $\ell$ and $m$.

\begin{figure}[t!]
\centering
\subfloat[Applying $x^2 y$ on the left]{%
  \includegraphics[width=1\textwidth]{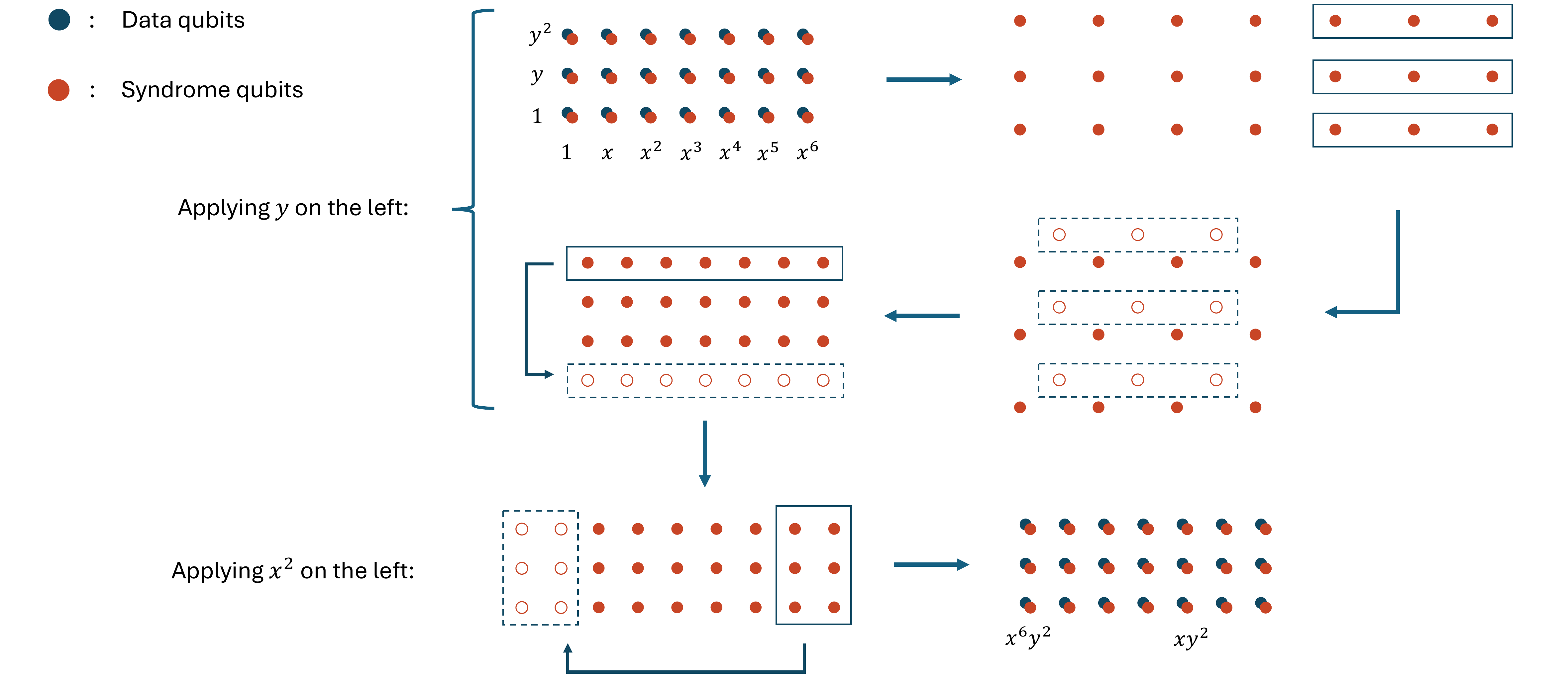}
}
\subfloat[Applying $x^2 y$ on the right]{%
  \includegraphics[width=1\textwidth]{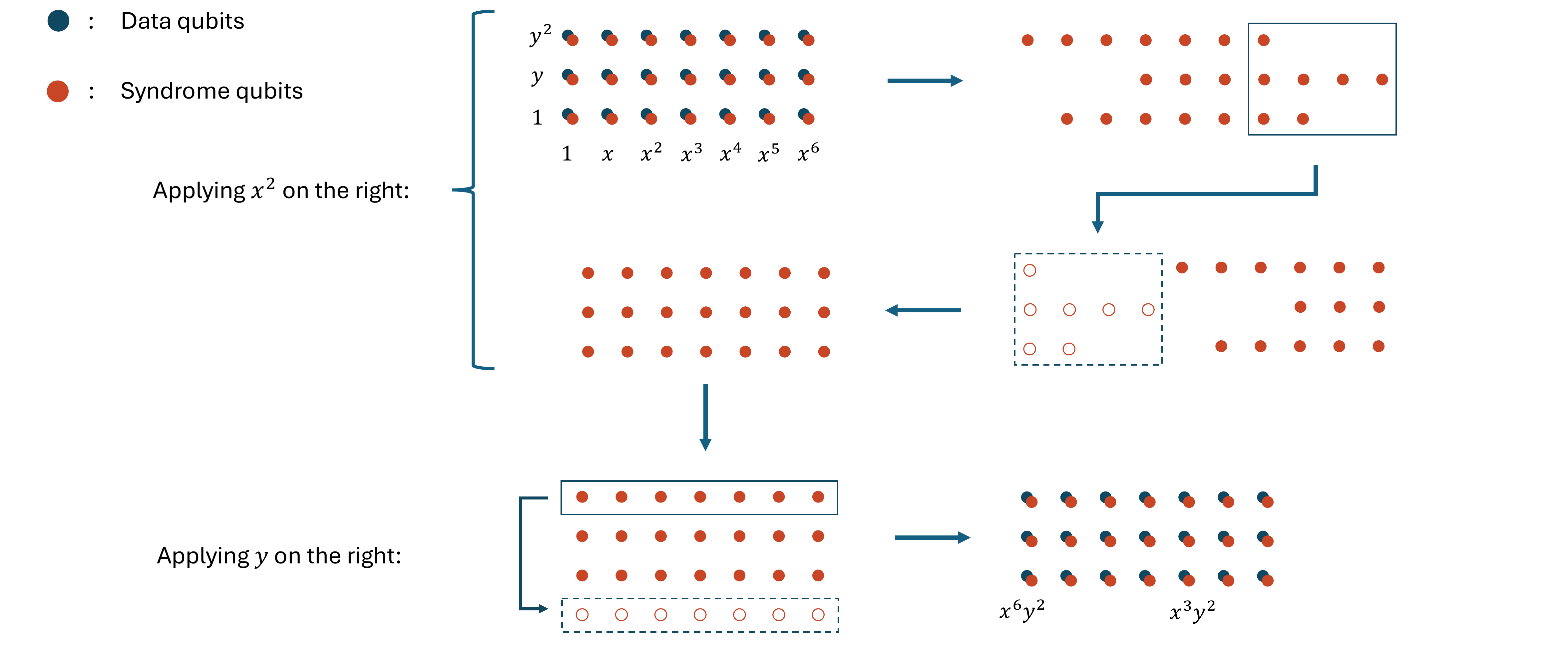}
}
\caption{The AOD moves required for implementing the monomial $x^2y$ on the horizontal data qubits for the group $\mathbb{Z}_7 \rtimes_2 \mathbb{Z}_3$ are depicted. The blue dots represent the horizontal data qubits, and the orange dots represent the syndrome qubits. 
\textbf{(a)} Applying $x^2y \in a$ on the left. The first and second steps illustrate the grid-type permutations (riffle shuffles) required for implementing $y$, and the later steps illustrate the vertical and horizontal cyclic shifts for the rest. In the last frame, the previous locations of two ancilla qubits are labeled. Applying $x^2y$ to their previous locations gives the present locations.
\textbf{(b)} Applying $x^2y \in b$ on the right. The first and second steps illustrate the distance shift for each row for implementing $x^2$. The last row is the horizontal cyclic shift for implementing $y$. The previous locations of two ancilla qubits are labeled.
}
\label{fig:implementation}
\end{figure}


\section{Discussion and outlook} \label{sec:outlook}

In this work, we introduced ZSZ codes as a non-abelian generalization of bivariate bicycle codes. We proposed a two-dimensional rectangular layout that facilitates the movement of ancilla qubits for syndrome extraction in neutral atom platforms. We then simulated memory experiments under a circuit-level depolarizing noise model using both conventional decoding with mid-circuit measurements as well as a local measurement-free scheme inspired by self-correcting memories. We observed numerical evidence of a sustainable threshold to the measurement-free scheme and found that it is higher than that of the 4D toric code under the same noise model and decoder. Finally, we described how to implement a complete round of syndrome extraction in neutral atom platforms with optical tweezer arrays.

There are several directions for future work that could improve both the ZSZ codes' decoding performance as well as their implementation in a fault-tolerant architecture. To start, it would be immediately beneficial if one can parallelize the $X$-syndrome and $Z$-syndrome extraction circuits to reduce the overall depth of syndrome extraction. A lower syndrome extraction depth would mean less noise accumulation for the data qubits due to idling as well as a faster QEC cycle time. It would also be practically relevant to optimize the CZ scheduling in order to mitigate the effects of hook errors, as has been done for the bivariate bicycle codes \cite{BB_codes}. In addition, for the measurement-free simulations, we only used a phenomenological noise model for the local decoding circuit. It would be more accurate to account for circuit-level noise in the local decoder, and this would most likely involve state-vector simulations due to the decoding circuit involving non-Clifford gates; one can also include noise during the atom routing process such as from heating, as has been done in \cite{Xu_2024_constant}. Alongside circuit-level optimizations, it would also be interesting explore other potential local decoding strategies, such as if local cellular automaton decoders for topological codes \cite{Kubica_2019_CA, balasubramanian2025} can be adapted to ZSZ codes. For the 4D toric code, it has been numerically observed that a local sweep rule outperforms the local majority vote in terms of both a higher threshold and lower logical error rates \cite{Breuckmann_2017}.

Although we have benchmarked the memory performance of ZSZ codes, there are several considerations that we have glanced over, which would be crucial to performing encoded logic. First, in our numerical simulations, we have assumed that the logical codespace has already been prepared fault-tolerantly. In an actual platform, we will typically need to prepare the codespace starting from a product state. The usual method for fault-tolerant state preparation of CSS codes works here: to initialize the logical $\ket{\overline{0}}$/$\ket{\overline{+}}$ state, we initialize all data qubits in $\ket{0}$ and perform $d$ rounds of syndrome measurement. Using spacetime mappings, one can also reduce this depth to $O(1)$ at the cost of $O(d)$ additional ancilla qubits \cite{Bergamaschi_2025, hong2024single, hillmann2024} per data qubit. For some single-shot codes such as the 4D toric code, the local dependencies among the check operators enable the constant depth without the additional ancillas \cite{Campbell_2019, Xu_2025_fast}. Since we discuss measurement-free error correction, it would be practical to also have a measurement-free version of state preparation. However, we are not yet aware of such a protocol, and progress along this direction would enable the full memory experiment to be measurement-free, other than the final readout measurement of all data qubits. Slow state preparation is okay if the purpose of the ZSZ code block is solely for memory and not computation, as we would in principle only need to prepare the codespace once at the beginning. However, if we are using state preparation as a subroutine for logical computation (e.g. in \cite{Xu_2025_fast}), then its speed may bottleneck the otherwise fast (single-shot) QEC cycles. In addition to logical state preparation, we would also require fault-tolerant gadgets to perform logical computation. Like for the BB codes, we can construct ancillary systems \cite{williamson2024surgery, swaroop2025adapters, he2025extractors, BB_codes, yoder2025tour} that can measure arbitrary combinations of logical Pauli operators within a code block or between different code blocks to realize the full logical Clifford group \cite{Litinski_2019}. Note that, unlike the BB codes, we cannot leverage translational symmetries (shift automorphisms) to reduce the size of this ancillary system \cite{yoder2025tour}. Furthermore, these ``lattice surgery'' schemes generically require $\mathrm{\Theta}(d)$ rounds of measurements and feedback for fault tolerance\footnote{See \cite{hillmann2024, aasen2025geo, aasen2025topo} for recent progress on single-shot lattice surgery for topological codes.}, and it is not yet known whether a measurement-free version exists.

\section*{Acknowledgments}

We thank Ali Fahimniya for enlightening discussions on neutral atom routing.  This work was supported by the Air Force Office of Scientific Research under Grant FA9550-24-1-0120 (JG, AL), the Office of Naval Research under Grant N00014-23-1-2533 (JG, AMK, AL), the Defense Advanced Research Projects Agency (DARPA) under Agreement No. HR00112490357 (YH), and the DoE ASCR Quantum Testbed Pathfinder program under awards No.~DE-SC0019040 and No.~DE-SC0024220 (YH).

\section*{Code availability}

All source code and data for the numerical simulations are available at \href{https://github.com/yifanhong/ZSZ-codes-numerics}{this GitHub repository}.


\appendix
\renewcommand{\thesubsection}{\thesection.\arabic{subsection}}

\section{Classical group codes} \label{app:group codes}

A classical linear code $\mathcal{C}$ encoding $k$ logical bits amongst $n$ physical bits is described by a $k$-dimensional subspace of $\mathbb{F}^n_2$, the vector space of all length-$n$ binary bitstrings. The $2^k$ bitstrings in this subspace are called logical codewords, and the code distance $d$ is defined as the minimum Hamming weight (i.e. number of ones) amongst all $2^k-1$ nonzero codewords. These three parameters are often packaged using the notation $[n,k,d]$. Because $\mathcal{C}$ is linear, we can choose $k$ codewords to form a basis for this vector subspace, which we can organize into a generator matrix $G \in \mathbb{F}^{k\times n}_2$ that succinctly describes $\mathcal{C}$. Using the dual vector space, we can equivalently describe $\mathcal{C}$ with a parity-check matrix $H \in \mathbb{F}^{m\times n}$ whose rows annihilate those of $G$, and the code is thereby defined as $\mathcal{C} = \ker H$. If the row and column weights of $H$ are bounded by constants independent of $n$, then we say that $H$ (and likewise $\mathcal{C}$) is a low-density parity-check (LDPC) code.

\subsection{Group theory}

We begin by reviewing some relevant group theory knowledge, the details of which can be found in textbooks \cite{dummit_foote}.

\begin{defn}[Cyclic group] The cyclic group $\mathbb{Z}_n$ obeys the presentation $\mathbb{Z}_n = \langle x | x^n = 1\rangle$.  It is an abelian group and, if $n$ is prime, has no nontrivial subgroups.
\end{defn}

\begin{prop} The automorphism group of $\mathbb{Z}_n$, denoted $\mathrm{Aut}(\mathbb{Z}_n)$, which corresponds to the group of all isomorphisms from $\mathbb{Z}_n$ into itself, is isomorphic to $\mathbb{Z}_n^\times$, which is the multiplication group of non-zero integers mod $n$, restricted to the integers coprime to $n$.
\end{prop}

\begin{defn}[Regular representation]
    Given a finite group $G$ of order $\abs{G}=n$, its \emph{left-regular representation} is the set of $n\times n$ permutation matrices $\{L_i\}, i=1,\dots, n$ where the basis runs over the group elements, and $(L_i)_{jk} = 1$ if and only if $g_j = g_i g_k$. The right-regular representation $(R_i)_{jk}$ is defined analogously using the right-multiplication rule $g_j = g_k g_i$.
\end{defn}

As an example, the regular representation of the cyclic group $\mathbb{Z}_n$ consists of circulant matrices $(L_i)_{jk} = \delta_{j,i+k}$.

\begin{defn}[Commutator subgroup]
    For a group $G$, its \emph{commutator subgroup} or \emph{derived subgroup} is defined as
    \begin{align}
        [G,G] := \langle\, [g,h]=g^{-1}h^{-1}gh \;|\; g,h\in G \,\rangle \, .
    \end{align}
\end{defn}

\begin{defn}[Derived series]\label{defn:derived series}
    For a group $G$, its \emph{derived series} is a sequence of groups
    \begin{align}
        G = G^{(0)} \triangleright G^{(1)} \triangleright G^{(2)} \triangleright \cdots
    \end{align}
    where $G^{(i)} := [G^{(i-1)},G^{(i-1)}]$ is normal in $G^{(i-1)}$. The smallest $l$ such that $G^{(l)}=1$ is trivial is known as the \emph{derived length} of $G$.
\end{defn}

As a simple example, the derived length for any abelian group is 1 since the its commutator subgroup is trivial. The notion of derived length will later prove useful when we analyze the expansion of Cayley graphs.

\subsection{Cayley graphs and group-algebra codes}

In this section, we present a geometrical perspective on classical group-algebra codes, which will aid in later arguments.

\begin{defn}[Cayley graph]\label{defn:cayley graph}
    Given a group $G$ and a set of generators $S\subset G$, the \emph{left Cayley graph} $\mathcal{G}(G,S) = (V,E)$ is a directed graph whose vertices $v_i \in V$ correspond to group elements $g_i \in G$, and $(v_i, v_j) \in E$ if and only if $g_j = s g_i$ for some $s \in S$. The \emph{right Cayley graph} is similarly defined but with the condition $g_j = g_i s$ instead.
\end{defn}

Note that if $\ident \in S$, then $G$ contains a self-loop at each vertex.

\begin{defn}[Graph double cover]\label{defn:double cover}
    Given a graph $\mathcal{G}=(V,E)$, the \emph{double cover} of $\mathcal{G}$ is the (possibly directed) bipartite graph $\bar{\mathcal{G}}=(V_1\cup V_2, E_1 \cup E_2)$ = $\mathcal{G} \times K_2$, where $K_2$ is the complete graph on two vertices. Specifically, $V_1$ and $V_2$ are copies of $V$, $(a_i\in V_1, b_j\in V_2) \in E_1$ if and only if $(v_i \in V,v_j \in V)\in E$, and $(b_j\in V_2, a_i\in V_1) \in E_2$ if and only if $(v_j,v_i)\in E$.
\end{defn}

In other words, $V_1$ and $V_2$ partition the left and right vertices respectively. Similarly, $E_1$ and $E_2$ partition the left-emanating and right-emanating edges. If the underlying base graph $\mathcal{G}$ is undirected, then $E_1 = E_2$. Since $\abs{V_1} = \abs{V_2} = \abs{V}$, the double cover $\bar{\mathcal{G}}$ is a balanced graph.

\begin{defn}[Group-algebra code; geometric interpretation]\label{defn:group-algebra code}
    Given a finite group $G$ with $\abs{G}=n$ and a set of generators $S$, the parity-check matrix of the \emph{left group-algebra code} is defined as the biadjacency matrix of $\bar{\mathcal{G}}_2 = (V_1\cup V_2, E_2) \subset \bar{\mathcal{G}}$, the double cover of the left Cayley graph $\mathcal{G}(G,S)$ with only right-emanating edges. Left and right vertices are mapped to bits and checks respectively. The parity-check matrix $H\in\mathbb{F}^{n\times n}_2$ is defined by
    \begin{align}
        H = \sum_{s\in S} L[s] \, ,
    \end{align}
    where $L[s] \in \mathbb{F}^{n\times n}_2$ denotes the left-regular representation of $s$. The right group-algebra code is defined analogously using right-actions.
\end{defn}

Note that any parity-check matrix associated with, say the left-regular representation of, an element of $\mathbb{F}_2[G]$ can be put in the form $H = L'\left(\ident + \sum_i L_i\right)$, which is equivalent to $H' = \ident + \sum_i L_i$ up to the global permutation (automorphism) $L'$. $H'$ is now the parity-check matrix of the Cayley graph code of $\mathcal{G}(G,\ident\cup\{S_i\})$ per the above definitions.

We now review some basic properties of a group-algebra code. Since the left and right degrees of a Cayley graph's double cover are both equal to the number of generators $\abs{S}$, if $\abs{S}$ is constant, then the associated group-algebra code is LDPC. Since the double cover has an equal number of bits and checks, the parity-check matrix of our Cayley graph code is square. Hence, the logical code dimension will be determined by the rank deficiency corresponding to linear dependencies among the parity checks.

\subsection{Classical ZSZ codes}

\begin{defn}[Semidirect product of two cyclic groups]
  Let $\ell,m$ be coprime integers and $q^m =1 \; \text{(mod $\ell$)}$.  The semidirect product of two cyclic groups can be presented as
  \begin{equation}
      \mathbb{Z}_\ell \rtimes_q \mathbb{Z}_m := \langle\, x, y | x^\ell=y^m=yxy^{-1}x^{-q}=1 \,\rangle .
  \end{equation}
\end{defn}

\begin{defn}[Classical ZSZ code]
We define a \emph{classical ZSZ code} to be a group-algebra code (Def. \ref{defn:group-algebra code}) formed from the group $\mathbb{Z}_\ell \rtimes_q \mathbb{Z}_m $ with parity-check matrix described by
\begin{equation}
    H = \sum_{i=1}^q \mathbb{B}\big[ x^{a_i} y^{b_i}\big] \, .
\end{equation}
where $\mathbb{B}[\cdot]$ can denote either the left-regular or right-regular representation.
\end{defn}


\section{Quantum 2BGA codes} \label{app:girth,diameter}

This appendix contains the relevant details for 2BGA code properties mentioned in the main text.

\begin{defn}[Two-block group-algebra (2BGA) code]
    Let $A,B \in \mathbb{F}^{n\times n}_2$ be the parity-check matrices of a left and a right group-algebra code (Def. \ref{defn:group-algebra code}) respectively based on the same group $G$ of order $n$. Then the CSS parity-check matrices of the quantum 2BGA code are given by
    \begin{subequations}\label{eq:2BGA check matrices}
    \begin{align}
        H_X &= \left(\, A \;\big|\; B \,\right)  \\
        H_Z &= \left(\, B^\transpose \;\big|\; A^\transpose \,\right)  \, .
    \end{align}
    \end{subequations}
\end{defn}

From \eqref{eq:2BGA check matrices}, it is clear that the check weights of a quantum 2BGA code are simply the sum of the check weights of the two component classical codes. The orthogonality of $H_X$ and $H_Z$ follows from the commutativity of $A$ and $B$, which follows from the associativity of group multiplication. \eqref{eq:2BGA check matrices} can also be interpreted as a tensor (hypergraph) product between $A$ and $B$ followed by factoring out the ``diagonal'' action $G \times G \rightarrow G$  \cite{Breuckmann_2021_BP}.

\subsection{Girth} \label{app:girth}

In this section, we provide upper bounds on the girths of the Tanner graphs of abelian and nonabelian 2BGA codes. Recall that a parity-check matrix $H$ generates the Tanner graph of the corresponding code. The Tanner graph is a type of bipartite factor graph that depicts how qubit nodes and check nodes are connected. Define the qubit $X$-adjacency ($Z$-adjacency) graphs as a $n$-vertex graph where vertex $i$ and $j$ are connected by an edge if and only if qubits $i$ and $j$ lie in the support of an $X$-check ($Z$-check). 
\begin{defn}[Graph girth]
    Given a simple graph $G=(V,E)$, let $\partial_{ve} \in \mathbb{F}^{\abs{V}\times\abs{E}}_2$ be the vertex-edge incidence matrix such that $\ker\partial_{ve}$ is the space of closed cycles or loops on $G$. The girth of $G$ is then defined as
    \begin{align}
        \mathrm{girth}(G) = \min_{c \in \ker\partial_{ve}} \abs{c} \, .
    \end{align}
\end{defn}
For a classical linear code $\mathcal{C}$ with parity-check matrix $H$, we define its code girth $\mathrm{girth}(H)$ as the girth of the bit adjacency graph corresponding to $H$. For a quantum CSS code, we define its girth to be the minimum girth between the qubit $X$-adjacency and $Z$-adjacency graphs.

Let us decompose the 2BGA block matrices $A$ and $B$ in terms of their components $A = \sum_i A_i$ and $B=\sum_i B_i$. Now let us consider the girth of the qubit $X$-adjacency graph induced by $H_X= \left(\, A \;\big|\; B \,\right)$. $H_X$ naturally divides the vertices into left and right sectors. If we start from a given left vertex, corresponding to some group element, all its left neighbors can be obtained by examining the associated column in $A_i^{\transpose} A_j$ ($i\neq j$) that corresponds to that group element; the right neighbors are obtained similarly but using $B_i^{\transpose} A_j$. Similarly, for right vertices, we examine $A_i^{\transpose} B_j$ and $B_i^{\transpose} B_j$ to get their left and right neighbors respectively.

\begin{thm}[Girth upper bound for abelian 2BGA codes]
    For an abelian 2BGA code with left and right sector weights at least 2, its code girth is at most 3.
\end{thm}

\begin{proof}
    For an abelian 2BGA code, recall that all component matrices $A_i$ and $B_i$ commute with each other. Consider the following path of length 3 on the $X$-type qubit adjacency graph, starting from a left vertex:
    \begin{align}
        P = \big(A_i^\transpose B^{}_k\big) \big(B_k^\transpose A^{}_j\big) \big(A_j^\transpose A^{}_i\big) \quad,\quad i \neq j \, .
    \end{align}
    We read the path $P$ above from right to left, and importantly the matrix type ($A$ vs $B$) must be the same when we go between the tuples in the parentheses, corresponding to arriving at a left or right vertex and subsequently leaving from that vertex. Since all the $A$s and $B$s commute, we can rearrange the above expression to get $P = \big(B^\transpose_k B^{}_k\big) \big(A^\transpose_j A^{}_j\big) \big(A^\transpose_i A^{}_i\big) = \ident$, using the fact that each $A$ and $B$ is an orthogonal matrix. Thus $P$ is a loop of length 3, and so the girth is at most 3.
\end{proof}

\begin{thm}[Girth upper bound for generic 2BGA codes]
    For any 2BGA code with left and right sector weights at least 2, its girth is at most 4.
\end{thm}

\begin{proof}
    Consider the following path of length 4 on the $X$-type qubit adjacency graph, starting from a right vertex:
    \begin{align}
        P = \big(B_l^\transpose A^{}_j\big) \big(A_i^\transpose B^{}_l\big) \big(B_k^\transpose A^{}_i\big) \big(A_j^\transpose B^{}_k\big) \quad,\quad i\neq j\;,\; l\neq k \, .
    \end{align}
    Rearranging the parentheses in the above expression, we get $P = B^\transpose_l \big(A^{}_j A_i^\transpose\big) \big(B^{}_lB_k^\transpose\big) \big(A^{}_iA_j^\transpose\big) B^{}_k$. Now, since $A^{}_j A_i^\transpose$ acts nontrivially solely on the left sector and likewise $B^{}_lB_k^\transpose$ on the right sector, they commute. Hence, upon rearranging, we have $P = B^\transpose_l \big(A^{}_j A_i^\transpose\big) \big(A^{}_iA_j^\transpose\big) \big(B^{}_lB_k^\transpose\big) B^{}_k = B^\transpose_l \big(A^{}_j A_i^\transpose A^{}_iA_j^\transpose\big) \big(B^{}_lB_k^\transpose\big) B^{}_k = B^\transpose_l B^{}_lB_k^\transpose B^{}_k = \ident$. Thus $P$ is a loop of length 4, and so the girth is at most 4.
\end{proof}

\subsection{Diameter} \label{app:diameter}

In this section, we show that the diameters of the Tanner graphs can be quite different between ZSZ codes and abelian 2BGA codes. We prove that if we fix the weight of parity checks $w$, we can construct families of ZSZ codes with diameter $O(\log n)$, where diameter is with respect to the qubit adjacency graphs. In contrast, for any BB code, the diameter scales as $\Omega(n^{1/(w-1)})$.

\begin{thm}[ZSZ Tanner graphs with small diameter]
    Suppose the block matrices $A$ and $B$ have weight 3. If $\ell = q^m-1$ and the component group elements $g_1,g_2,g_3$ obey $g_2 g^{-1}_1 g_3 \neq g_3 g^{-1}_1 g_2$, then the corresponding ZSZ code has diameter $O(\log n)$.
\end{thm}

\begin{proof}
    We will analyze the case where the Tanner graphs of $A$ and $B$ are connected and deal with disconnected components afterward. In the connected case, the diameter of $H_X$ is less than the sum of the diameters of the graphs of $A$ and $B$. Our goal is to show that the diameters of $A$ and $B$ can be $O(\log n)$. Without loss of generality, we will only discuss the diameter of the graph of $A$.

    Every vertex (qubit) in the left sector can be represented as $x^i y^j$, and all vertices are equivalent to one another due to the transitivity of the group. Let us start with the simplest case where $s=\{1, x, y\}$. We will show that all vertices can be reached from the vertex $x^0 y^0$ in $O(m) = O(\log n)$ steps. In each step, the vertices that can be reached from any initial vertex can be obtained by applying $\{x,x^{-1},y,y^{-1}, x^{-1}y, y^{-1}x\}$. For $q>1$, any positive integer $i<q^m$ can be decomposed as $i=\sum_{a=0}^{m-1} i_a\, q^a$, where $0\leq i_j<q$. Since $yx=x^q y$, we have $x^i y^j = \prod_{a=0}^{m-1} (x^{i_a} y)\, y^{j+1}$, which means that the vertex $x^i y^j$ can be reached in $\sum_a i_a +m+j < \left((q+1)/2+2\right)m$ steps, and hence the diameter is bounded by $\left((q+1)/2+2\right)m = O(\log n)$.
    
    Now consider $s=\{1, x^u, x^{u'} y^v\}$ and $\text{gcd}(u,\ell) = \text{gcd}(v,m)=1$. The latter condition is to ensure that the graph is connected. We then have $(x^{u'}y^v) x^u = x^{u q^v} (x^{u'} y^v)$, which says that this graph is isomorphic to one generated by $s'=\{1, x, y\}$ if we suitably modify $q \rightarrow q' = q^v$. For $s'$ and $q'$, we can decompose $i<q^m$ as $i=\sum_{a=0}^{m-1} i_a\, q^{va\, \text{mod}\, m} = \sum_{a=0}^{m-1} i_a\, q'^a$, and hence the diameter has the same upper bound $\left((q+1)/2+2\right)m = O(\log n)$.
    
    For a more general set of generators $s=\{s_1, s_2 , s_3\}$, we can decompose $s$ as $s = s_1 s'  =s_1  \{1,s_1^{-1} s_2 , s_1^{-1} s_3 \} = s_1 \{1, s'_2 , s'_3\} $. Since the adjacency matrix is generated by $s^{-1} s_j \backslash \{1\} = s'^{-1}_i s'_j \backslash \{1\}$, to $x^{i'} y^{j'}$, the graph generated by $s$ is isomorphic to the that of $s'$. 
    The interpretation of the condition that the generators have to satisfy becomes clear now: $s_2 s^{-1}_1 s_3 \neq s_3 s^{-1}_1 s_2$ means $s'_2$ and $s'_3$ don't commute, so we can get $x^u = s'_2 s'_3 s'^{-1}_2 s'^{-1}_3 $ with a nonzero $u$. The above condition also implies $s'_2$ or $s'_3$ has a nonzero exponent of $y$. If we assume $s'_2 = x^{u'} y^{v}$ with $\text{gcd}(v,m)=1$ and $\text{gcd}(u,n)=1$, the situation becomes similar to the previous case and we can get a looser upper bound of the diameter $\left(2(q+1)+2\right)m = O(\log n)$.
    
    So far, we assumed the graph of $A$ is connected, which corresponds to the condition $\text{gcd}(u,\ell) = \text{gcd}(v,m)=1$. If the graph is disconnected, we start with the fact that the diameter of $H_X$ is less than the sum of every disjoint subgraph of $A$ and $B$. Consider the worst case, where $\text{gcd}(u,\ell)=u$ and $\text{gcd}(v,m)=v$. The graph can be separated into $uv$ disconnected subgraphs. It can be checked that each subgraph is equivalent to the graph of $s=\{1, x, y\}$ with $q \rightarrow q^v$, $s \rightarrow s/u$ and $m \rightarrow m/v$. Therefore, the diameter of each connected graph is still $O(m)$ and the sum of the diameters is $O(uvm)$. As a result, for any nonzero $u$ and $v$, the diameter of $H_X$ is bounded by $O(\log n)$.
\end{proof}

\begin{thm}[Diameter of abelian 2BGA Tanner graphs]
    Suppose we have a 2BGA code with check weight $w_c$, i.e. there are $w_c$ total group generators in the specifications of the block matrices $A$ and $B$. Then the diameter of the $X$-type and $Z$-type Tanner graphs is at least $\mathrm{\Omega}\big( n^{1/(w_c-1)} \big)$.
\end{thm}

\begin{proof}
    For abelian 2BGA codes, the diameter of the Tanner graph corresponding to $H_X$ or $H_Z$ is larger than that corresponding to $A+B$. Hence, a lower bound on the diameter of $A+B$ will also be a lower bound on the diameters of $H_X$ and $H_Z$. Recall that we can apply a change of basis to all generators to normalize one of them to be 1. We can then regard the $w_c-1$ nontrivial generators as unit vectors that span a vector space where data qubits are associated with distinct vectors. Since we have $w_c-1$ unit vectors, this vector space is at most $(w_c-1)$-dimensional, which means that the diameter of the Tanner graph is at least $\mathrm{\Omega}\big(n^{1/(w_c-1)}\big)$.
\end{proof}


\section{Self-correction, confinement and expansion}

This appendix contains a brief primer on self-correcting memories and expander graphs. Given a classical linear or quantum stabilizer code, we can define a code Hamiltonian that is a (negative) sum of all parity checks. This Hamiltonian is fully commuting and has an integer spectrum labeled by distinct error syndromes whose ground state subspace is the codespace. Loosely speaking, we say that a code is self-correcting if we allow it to interact with a heat bath, according to its code Hamiltonian, at some constant nonzero temperature and can successfully recover the encoded information with high probability after this interaction for a time diverging with system size. Typically, one also assumes that the system-bath interaction is local and obeys detailed balance, so that the steady state is given by the Gibbs distribution; examples of such interactions include classical Markov-chain Monte Carlo algorithms \cite{Metropolis_1953, Hastings_MC_1970, Glauber_1963} as well as their quantum generalizations \cite{Temme_QuantumMetropolis_2011, Chen_Exact_2023, Ding_Efficient_2025, Jiang_QuantumMetropolis_2024}. Successful final recovery is dependent on a specific decoder, which is usually assumed to be a theoretically tractable decoder such as minimum-weight or maximum-likelihood. Formally, let $\rho_\mathcal{C}$ denote an arbitrary state in the codespace, $\mathcal{B}_{T,\tau}$ the CPTP map for the interaction with the bath at temperature $T$ for time $\tau$, and $\mathcal{R}$ the CPTP map for the final recovery operation. We say that a code is self-correcting if for $T < T_{\rm c} = \mathrm{\Omega}(1)$ and $\tau = \omega_n(1)$ (typically exponential),
\begin{align}\label{eq:self-correcting defn}
    \mathcal{R} \circ \mathcal{B}_{T,\tau}[\rho_\mathcal{C}] \propto \rho_\mathcal{C}
\end{align}
with probability $\mathbf{P} = 1 - o_n(1)$. In other words, the Knill-Laflamme QEC conditions \cite{KnillLaflamme_2000} are satisfied for the error channel corresponding to $\mathcal{B}_{T,\tau}$. The word ``self-correcting'' originates from the interpretation that the bath simultaneously corrects as well as produces errors, and below the critical temperature there is a strong bias towards correction. There will always be some residual error at the end, and so the final recovery $\mathcal{R}$ is more or less just a proxy for $\mathcal{B}_{T,\tau}$ according to \eqref{eq:self-correcting defn}.

An important quantity in the analysis of a self-correcting memory is the free energy $F = E - TS = -T\log Z$, where $S$ is the entropy and $Z$ is the canonical partition function. The free energy characterizes the competition between entropy (related to errors) and energy (related to correction). Intuitively, for transitions between different states given by the system-bath interaction, the probability is proportional to $\mathrm{e}^{-\beta\Delta E}$, and the number of transitions that are accessible is proportional to $\mathrm{e}^{\Delta S}$, and so transitions that increase the free energy are unfavored whereas those that decrease the free energy are favored in typical trajectories of the system under random dynamics. Hence, if we desire a lower bound on the self-correcting memory time $\tau$, then it suffices to obtain a lower bound on the free energy cost of incurring a logical error: in other words, to argue that there is a low free energy bottleneck around each codeword that must be reached to incur a logical error \cite{thermal_LDPC}.

Confinement is a code property which, loosely speaking, asserts that increasingly larger errors produce increasingly larger syndromes (energy cost), up to a cutoff known as the energy barrier of the code. Beyond this energy barrier, a suitably large error may actually produce a small syndrome \cite{thermal_LDPC,rakovszky2024ii}. The energy cost as a function of the error weight is known as the confinement function. For an LDPC code, since each qubit only participates in a constant number of parity checks, the confinement function can at most be linear. It was recently shown that a linear confinement function is sufficient to overcome any effects of entropy and result in a macroscopic free energy barrier, thus leading to self-correction \cite{thermal_LDPC,Placke:2024wey, gamarnik2024slow, rakovszky2024bottlenecks}.

\begin{thm}[Small-set linear confinement implies self-correction (informal) \cite{thermal_LDPC}] \label{thm:linear confinement self-correction}
    Suppose we have a quantum LDPC stabilizer code of length $n$ such that the energy cost $\mathcal{E}(P)$ of Pauli error $P$ satisfies
    \begin{align}\label{eq:linear confinement}
        \mathcal{E}(P) \geq \alpha \abs{P} \;,\;\; \forall\, \abs{P} \leq \gamma n^b
    \end{align}
    for constants $\alpha,\gamma,b>0$ independent of $n$. Then there exists a critical temperature, below which the self-correcting memory time $\tau$ diverges as
    \begin{align}
        \tau = \mathrm{e}^{\mathrm{\Omega}(n^b)} \, .
    \end{align}
\end{thm}

Note that self-correcting memories exist without the strict requirement of linear confinement \eqref{eq:linear confinement} such as toric codes in $D\geq4$ spatial dimensions \cite{Alicki_2010} and color codes in $D\geq6$ spatial dimensions \cite{Bombin_2013_6D}. However, in both of these examples, the parity checks form an extensively overcomplete set, and the special redundant structure restricts entropic effects. Since our ZSZ codes do not possess such an extensive amount of redundant parity checks, we speculate that our numerical observations of self-correction are attributed to strong confinement in typical error clusters. From \eqref{eq:linear confinement}, it is clear that linear confinement requires the size of the neighborhood or boundary of small vertex sets to be proportional to their volume. This boundary $\propto$ volume correspondence is a defining trait of a special class of sparse graphs known as expander graphs.

\begin{defn}[Graph expansion]\label{defn:expander graph}
    For a graph $\mathcal{G}=(V,E)$, its \emph{edge expansion} or \emph{Cheeger constant} $h(\mathcal{G})$ is defined as
    \begin{align}\label{eq:cheeger}
        h(\mathcal{G}) = \min_{S\subset V:\, \abs{S}\leq\abs{V}/2} \frac{\abs{\partial S}}{\abs{S}} \, ,
    \end{align}
    where $\partial S := \big\{ \{u,v\}\in E \,:\, u\in S,\, v\notin S \big\} \subset E$ is the edge neighborhood of $S$. We say a family of graphs $\{\mathcal{G}_i\}$ of increasing sizes is \emph{expanding} if $h(\mathcal{G}_i) = \mathrm{\Omega}(1)$.
\end{defn}

Intuitively, the edge expansion $h(\mathcal{G})$ tells us how many edges we need to cut in order to disconnect the graph into disjoint components. Note that the condition $h(\mathcal{G}) = \mathrm{\Omega}(1)$ implies that $\mathrm{diam}(\mathcal{G}) = O(\log n)$, since the individual neighborhoods of any two vertices grow exponentially and eventually encompass at least half of all vertices. By the pigeonhole principle, there must then exist a path of length $O(\log n)$ between any two vertices. However, the converse is not true. A regular tree has logarithmic diameter but poor expansion because cutting any edge disconnects the entire branch attached to it, which may have extensive size. 

Typically, one can derive linear confinement bounds from the underlying (bipartite) expansion of the code's Tanner graph. If the expansion coefficient of a subset is larger than half the vertex degree, then there must exist a fraction of parity checks in the neighborhood of that subset that are only connected by a single edge, i.e. a ``unique neighbor'', and hence correspond to unsatisfied checks if the subset is the support of an error. For random graphs, one can probabilistically demonstrate this unique-neighbor expansion up to a constant fraction of vertices using combinatorial arguments \cite{Gallager_1962, ModernCodingTheory}. Explicit constructions typically revolve around local modifications of expanding Cayley graphs, e.g. \cite{Lubotzky_1988, Margulis_1988}, in order to upgrade them to unique-neighbor expanders \cite{Tanner_1981, Sipser_1996, Capalbo_2002}. We now recite a no-go theorem regarding the expansion of Cayley graphs.

\begin{thm}[Non-expanding Cayley graphs; Corollary 3.3 of \cite{Lubotzky_1992}]
\label{thm:non-expansion derived length}
    Let $\{G_i\}$ be a family of finite groups, each with derived length $l=O(1)$ and generating set $S_i$ with $\abs{S_i} = O(1)$. Then the family of Cayley graphs $\mathcal{G}(G_i,S_i)$ is not expanding.
\end{thm}

We will not formally rederive this established result, but we can explain the intuition behind it. Loosely speaking, a group with a constant derived length, also known as a solvable group, has an ``almost abelian'' structure in the sense that it can viewed as finite extensions of abelian groups. We can see this ``almost abelian'' structure inside $\mathbb{Z}_\ell \rtimes_q \mathbb{Z}_m$ as follows. Its commutator subgroup is precisely the normal subgroup $\mathbb{Z}_\ell$. When we quotient out this normal subgroup, we obtain $\mathbb{Z}_m$. Geometrically (recalling the rectangular layout of Figure \ref{fig:X-check Cayley graphs}), we are treating each row (copy of Cay($\mathbb{Z}_\ell, x$)) as a conglomerate object, with the automorphism $\varphi(x)=x^q$ acting as a coarse-grained ``edge'' between neighboring rows; this is the structure of Cay($\mathbb{Z}_m, y$), which is not an expander. In particular, the size of the neighborhood of a collection of adjacent rows remains constant irrespective of the number of included rows.

\begin{cor}[ZSZ derived length]
    The derived length of the group $\mathbb{Z}_\ell \rtimes_q \mathbb{Z}_m$ is 2, and thus a ZSZ Cayley graph is not an expander according to Definition \ref{defn:expander graph}.
\end{cor}

\begin{proof}
    Without loss of generality, we label all elements of $\mathbb{Z}_\ell \rtimes_q \mathbb{Z}_m$ as $x^i y^j$ for $i=1,\dots,\ell$ and $j=1,\dots,m$. Since both $\mathbb{Z}_\ell$ and $\mathbb{Z}_m$ are abelian, the only nontrivial elements in the commutator subgroup of their semidirect product come from terms between the two groups. For generators $x \in \mathbb{Z}_\ell$ and $y \in \mathbb{Z}_m$,
    \begin{align}
        [x,y] = x^{-1}y^{-1}xy = x^{-1}x^{q^{m-1}}y^{-1}y = x^{q^{m-1}-1}
    \end{align}
    which generates an abelian subgroup of $\mathbb{Z}_\ell$. Since the commutator subgroup of an abelian group is trivial, the derived series of $\mathbb{Z}_\ell \rtimes_q \mathbb{Z}_m$ terminates after two iterations. Theorem \ref{thm:non-expansion derived length} then asserts that this family of Cayley graphs cannot be expanding.
\end{proof}

Notice, however, that Definition \ref{defn:expander graph} for an expander graph is a \emph{global} one: the boundary $\propto$ volume scaling must persist to half the size of the graph. As is clear from Theorem \ref{thm:linear confinement self-correction}, quantum self-correction will follow from any LDPC code with sufficient expansion on small sets, whose maximum size can be subextensive $o(n)$. Due to both the strong numerical evidence of self-correction as well as the lack of local metachecks, we speculate:

\begin{conj}[Free energy barrier in ZSZ codes]
    There exists a particular sequence of ZSZ codes of increasing length $n$, which exhibit a free energy barrier that grows with $n$ around codewords, and as such possess a self-correction threshold.
\end{conj}

\subsection{Relation to noisy greedy decoding} \label{app:glauber}

In this section, we briefly review the connection between local greedy decoding and a particular (classical) Gibbs sampler known as Glauber dynamics. For a quantum CSS code with code Hamiltonian $\mathcal{H} = \mathcal{H}_X + \mathcal{H}_Z$, classical Glauber dynamics on independent $X$ and $Z$ errors is sufficient to sample the entire spectrum of the code Hamiltonian and hence is a valid quantum Gibbs sampler. Without loss of generality, we focus our attention on $X$ errors. In Glauber dynamics, we select a data qubit (labeled by $j$) of the code at random and measure its connected $Z$-checks to obtain a local error syndrome. We then apply the local Pauli $X_j$ operator with probability
\begin{align}\label{eq:Glauber rule}
    \mathbf{P}\big( \ket{\psi}\rightarrow X_j\ket{\psi} \big) = \frac{1}{1+\mathrm{e}^{\beta\mathrm{\Delta} E}} \, ,
\end{align}
where $\mathrm{\Delta} E$ is the change in energy or syndrome weight upon applying $X_j$, and $\beta$ is the inverse temperature. One can quickly verify that \eqref{eq:Glauber rule} along with its $Z$-error version satisfy the detailed balance condition $\rho(\ket{\psi'}) \mathbf{P}(\ket{\psi'}\rightarrow\ket{\psi}) = \rho(\ket{\psi}) \mathbf{P}(\ket{\psi}\rightarrow\ket{\psi'})$ when $\rho \propto \mathrm{e}^{-\beta\mathcal{H}}$ is the equilibrium Gibbs state. Since every data qubit has an equal probability of being chosen, every eigenstate of $\mathcal{H}$ has a nonzero probability of being reached (ergodicity). Thus, the Gibbs state is the unique steady state of Glauber dynamics \eqref{eq:Glauber rule}.

To understand the connection of \eqref{eq:Glauber rule} to greedy decoding, let us first examine the zero-temperature $\beta\rightarrow\infty$ limit. In this limit, \eqref{eq:Glauber rule} becomes a step function: local Paulis that lower the energy ($\mathrm{\Delta}E>0$) are always applied, those that increase the energy ($\mathrm{\Delta}E<0$) are never applied, and ties ($\mathrm{\Delta}E=0$) are handled according to a 50\% coin toss. We can interpret this limit as a greedy decoder that tries to always locally lower the energy according to a majority vote. When $\beta$ is finite, the greedy decoder now applies the local majority ``correctly'' with probability $p$ \eqref{eq:Glauber rule} and ``incorrectly'' with probability
\begin{align}\label{eq:q Glauber}
    q(\mathrm{\Delta}E) = 1-p(\mathrm{\Delta}E) = \frac{1}{1+\mathrm{e}^{\beta\abs{\mathrm{\Delta} E}}} \, .
\end{align}
\eqref{eq:q Glauber} may seem a bit odd at first since it requires very specific failure probabilities for different inputs to the majority vote. It could be the case that we have a uniform noise model where the greedy decoder outputs the wrong answer with a fixed probability for all inputs. In this case, we cannot exactly map it to Glauber dynamics, but we can provide an upper bound and say that it fails \emph{less often} than Glauber dynamics at some finite temperature. Since our code is LDPC, each qubit participates in at most $w_q = O(1)$ $Z$-checks. Notice that at fixed $\beta$, \eqref{eq:q Glauber} is minimized when $\abs{\mathrm{\Delta} E}$ is maximized, which occurs when either all or none of the parity checks are satisfied:
\begin{align}\label{eq:q_min Glauber}
    q_{\rm min} = \frac{1}{1+\mathrm{e}^{\beta w_q}} \, .
\end{align}
So given some failure probability $q$, we can provide an upper-bound temperature $\beta_{\rm upper}$ according to \eqref{eq:q_min Glauber}. If $\beta_{\rm upper} > \beta_{\rm c}$ is in the self-correcting phase, then we can expect the performance of the noisy greedy decoder to be \emph{no worse} than that of Glauber dynamics at temperature $\beta_{\rm upper}$. Note that we can always deliberately add in our own ``errors'' to exactly match \eqref{eq:q Glauber} if desired.


\section{Syndrome extraction and measurement-free decoding}

\subsection{Single-ancilla syndrome extraction}\label{app:syndrome extraction}

\begin{figure}[t]
    \centering
    \includegraphics[width=0.75\textwidth]{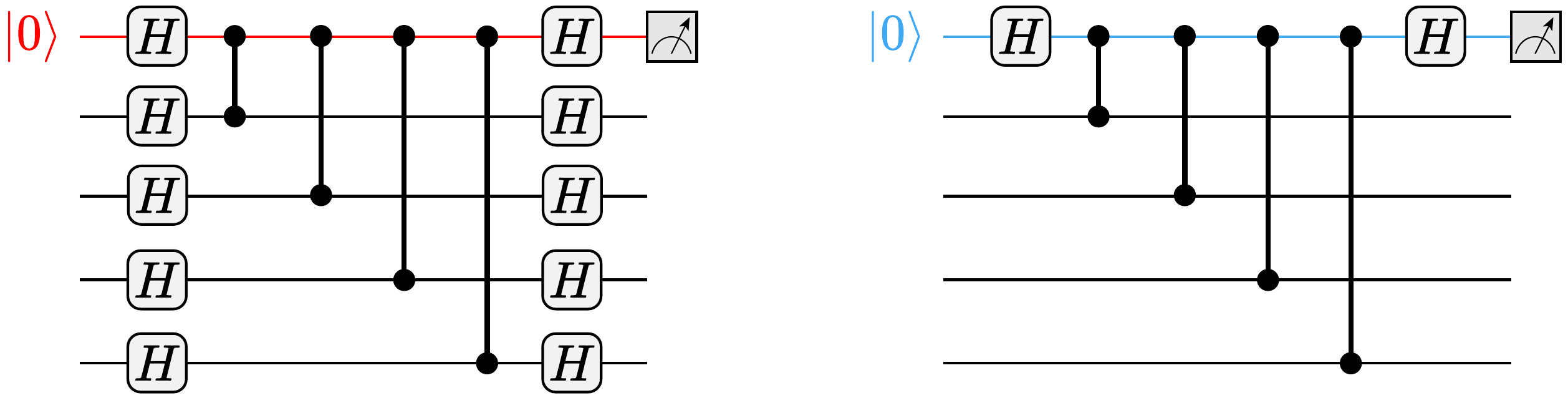}
    \caption{Single-ancilla syndrome extraction circuits for four-qubit $X$-check (left) and $Z$-check (right) operators are shown. The top qubit denotes the syndrome qubit that is to be initialized and measured to obtain the parity of the associated check operator.}
    \label{fig:syndrome ext circuits}
\end{figure}

Since ZSZ codes are LDPC, they are amenable to single-ancilla syndrome extraction: we initialize an ancilla ``syndrome'' qubit for each parity check and interact the syndrome qubit with its corresponding data qubits according to the circuits in Figure \ref{fig:syndrome ext circuits}. While simple, the one drawback to this approach is the potential introduction of ancillary hook errors, where a single fault on the ancilla in the middle of the circuit can propagate to multiple faults on its data qubits. However, since our ZSZ codes are LDPC with check weight 6, these circuits have depth 6 and so any hook error has support on at most 3 data qubits\footnote{More precisely, all hook errors are stabilizer-equivalent to an error with weight at most 3. Note that a hook error occurring at the beginning of the circuit enacts the check operator itself which is trivial since it belongs to the stabilizer group.}.
Thus, up to constant-weight hook errors, this syndrome extraction strategy is inherently fault-tolerant for ZSZ codes. To further mitigate the effect of ancillary hook errors, one can either schedule the CZ gates in a particular order or employ flagged syndrome extraction \cite{Chao_2018_flag} to detect the presence of hook errors at the cost of additional ancilla qubits.

For neutral atoms, the Hadamard gate can be realized by single-qubit laser pulses, and a high-fidelity CZ gate between two atoms can be realized by a blockade interaction \cite{Jaksch_2000, Lukin_2001}. A full round of single-ancilla syndrome extraction can then be performed in two stages as follows. In the first stage, we will extract the $X$-syndrome. Initialize $\ell m = n/2$ syndrome qubits all in $\ket{0}$ and apply a Hadamard gate to all syndrome and data qubits. Using the AOD optical tweezers, move each syndrome qubit to its corresponding data qubits and apply the necessary CZ gates. Finally, apply Hadamard to all syndrome and data qubits. Up to physical noise, the syndrome qubits now ``store'' the classical $X$-syndrome corresponding to $Z$ errors. To extract the $Z$-syndrome in the second stage, we perform a nearly equivalent procedure except that we do not apply the Hadamard gates to the data qubits.

\subsection{Single-shot greedy decoder implementation}\label{app:passive decoder}

At the heart of our passive memory lies a single-shot implementation of the local greedy decoder. A decoder is single-shot if it can reliably operate using information from only $O(1)$ syndrome extraction cycles \cite{Bombin_2015}. Roughly speaking, the goal of a single-shot decoder is not to eliminate \emph{all} errors, but only to reduce enough errors at each cycle such that the encoded information can \emph{eventually} be recoverable. For scalable fault-tolerant MFQEC, it is crucial that this decoder be not only single-shot but also implementable with a constant-depth circuit. So in addition to requiring $O(1)$ syndrome extraction cycles, we also demand $O(1)$ classical time complexity, so that qubit idling times remain finite even as $n\rightarrow\infty$. Note that typical single-shot decoders only require the first condition but not the second; e.g. MWPM for fixing broken loop excitations in the 3D toric code has an $O(n^3)$ classical time complexity.

It remains to show that we can implement a full sweep (all sites acted on once) of the greedy decoder in constant depth. A naive scheduling of the sweep such as ``typewriter'' order may require linear depth since the majority vote on the next data qubit may require the updated syndrome information from the majority vote on the previous qubit. Fortunately, because the ZSZ codes are LDPC, there exist ``non-overlapping'' sets of data qubits that we can address in parallel. Define the qubit $X$-adjacency graph as the simple graph where vertices denote data qubits, and two vertices are connected by an edge if and only if their associated data qubits share an $X$-check; the qubit $Z$-adjacency graph follows suit but with respect to the $Z$-checks. Since each data qubit participates in 3 $X$-checks, and each $X$-check involves $6$ data qubits, the maximum degree of the qubit $X$-adjacency graph is $3(6-1)=15$. Brooks's theorem \cite{Brooks_1941} then tells us that we can partition all the vertices into at most 15 non-overlapping subsets from which we can apply our greedy decoder in parallel. We note that the vertex degree is only an upper bound, and for some graphs the minimum partition, or chromatic number, can be smaller. Using ``sequential'' and ``independent-set'' greedy graph coloring algorithms in \textsf{NetworkX}, we found colorings between 7 and 10 for all of our ZSZ codes in Table \ref{tab:ZSZ codes}. In general, it is \textsf{NP}-hard to compute the chromatic number $\chi$ of an arbitrary graph. Nonetheless, Brooks's theorem provides us with an upper bound $\chi=O(1)$ for LDPC codes. Perhaps better approximate colorings can be obtained by leveraging the underlying Cayley graph structure of the qubit adjacency graphs.

Suppose we have a $\chi$-coloring of our qubit adjacency graph with $\chi=O(1)$. Within each subset of fixed color, we can apply the greedy decoder to all qubits in parallel. The full decoding procedure then proceeds as follows:
\begin{enumerate}
    \item Perform a round of ($X$ or $Z$) syndrome extraction with $n/2$ syndrome qubits.
    \item Choose a non-overlapping subset of data qubits according to the $\chi$-coloring of the qubit adjacency graph and apply the greedy decoder to all subset data qubits in parallel. Update the relevant bits in the error syndrome.
    \item Iterate step 2 ($\chi$ total iterations) until all data qubits have been addressed by the greedy decoder.
\end{enumerate}
In a classical computer, steps 2 and 3 can be compiled into simple boolean arithmetic in a straightforward manner. For a MFQEC implementation, we will need to compile these instructions into a quantum circuit. For step 2, we can borrow the reversible circuit of \cite{Boykin_2005} to implement the majority vote on three inputs:
\begin{align}
 \begin{quantikz}
     \lstick{$\ket{q_0}$} & \ctrl{1} & \ctrl{2}& \targ{} & \rstick{$\ket{\text{maj}}$} \\
     \lstick{$\ket{q_1}$} & \targ{} & & \ctrl{0} & \ground{} \\
     \lstick{$\ket{q_2}$} & & \targ{} & \ctrl{-2} & \ground{}
 \end{quantikz}
\end{align}
The above majority circuit does not preserve the inputs, but one can simply copy their classical values prior to the circuit using fresh ancillas in $\ket{0}$ and CNOT gates. After the majority qubit is obtained, we apply a controlled gate from it to the data qubit, with the target rotation being $X$ or $Z$ depending on which Pauli error we are correcting. The syndrome qubits can also be updated similarly using CNOT gates from the majority qubit. These updated syndrome qubits can then be used for step 3. After we perform the full sweep of the greedy decoder on all data qubits, we can discard all ancilla qubits and repeat the process, starting with a new round of syndrome extraction.


\section{Additional numerical simulations}

\subsection{Sustainable threshold estimation}\label{app:passive threshold}

We perform additional numerical simulations of passive error correction near the region ($p\approx 10^{-3}$) where the curves intersect in Figure \ref{fig:passive decoding}.
Figure \ref{fig:threshold estimate} presents these numerical results for both the ZSZ codes as well as the 4D toric code family. From the plots, we estimate the sustainable threshold of ZSZ codes to be $p^{\rm ZSZ}_{\rm th} \approx 0.095\%$; similarly, we estimate a sustainable threshold of $p^{\rm 4D}_{\rm th} \approx 0.063\%$ for the 4D toric code.

\begin{figure}[t]
    \centering
    \includegraphics[width=0.48\textwidth]{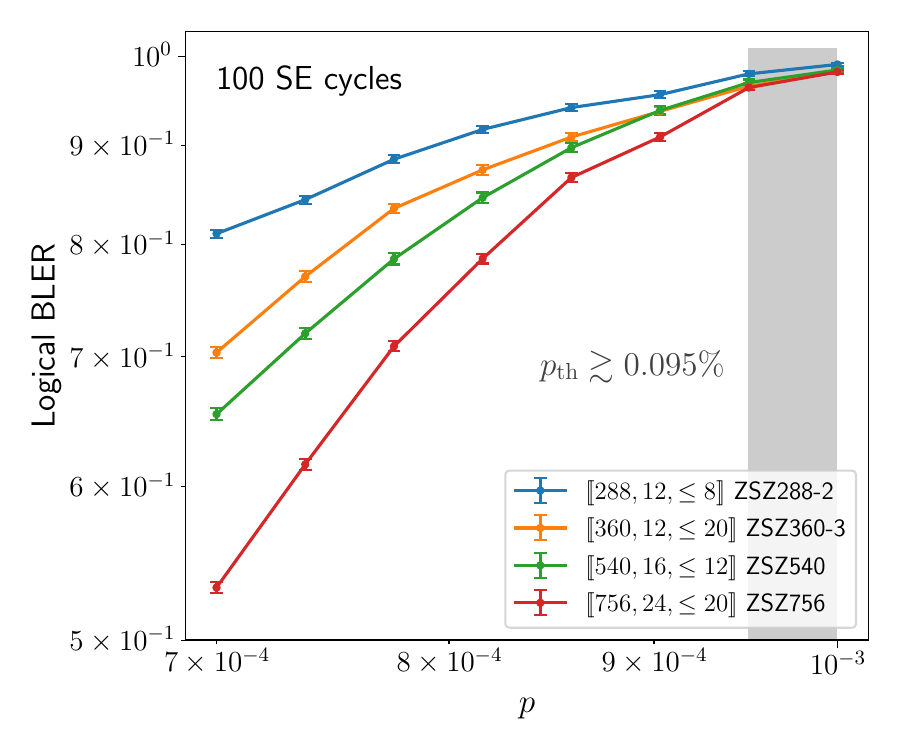} \hfill
    \includegraphics[width=0.48\textwidth]{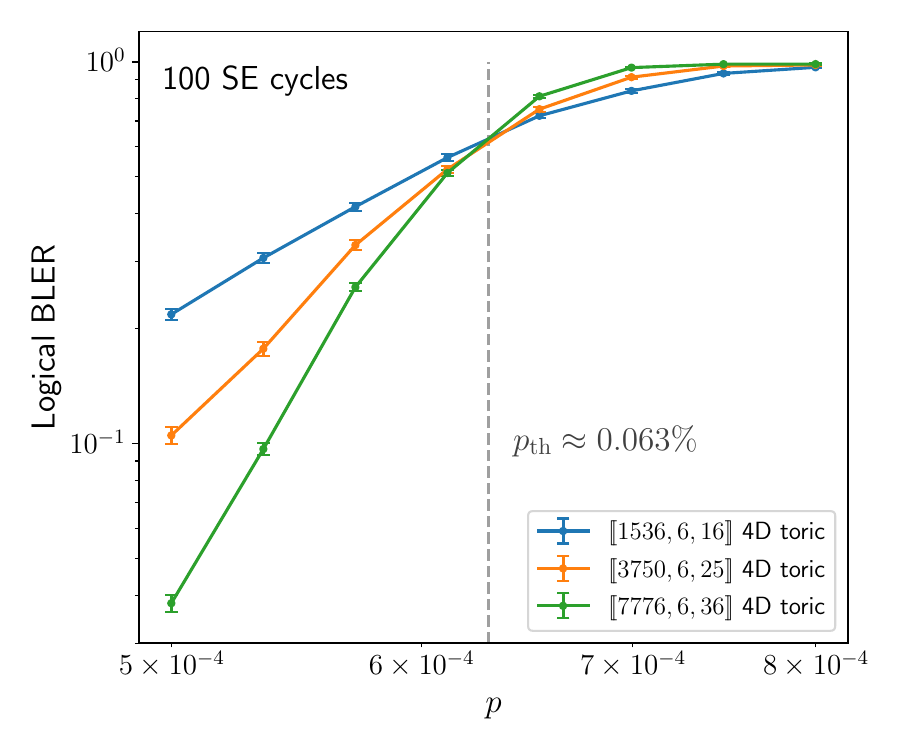}
    \caption{The logical block error rate (BLER) as a function of the physical noise strength $p$ is plotted near the observed sustainable thresholds for both the ZSZ codes as well as 4D toric codes with linear sizes $L=4,5,6$ under passive decoding. The curves of the 4D toric codes display a clear intersection around $p\approx 0.063\%$. The curves of the ZSZ codes do not show an obvious intersection, but we observe subthreshold behavior (decreasing BLER with increasing $n$) below the region highlighted by the shaded rectangle, and so we estimate a lower bound $p\gtrsim 0.095\%$ on the sustainable threshold.}
    \label{fig:threshold estimate}
\end{figure}

\subsection{Passive error correction of bivariate bicycle codes}\label{app:passive BB}

We perform an equivalent numerical search and noise simulation for BB codes as for our ZSZ codes under passive decoding. See Table \ref{tab:BB codes} for the parameters of the BB codes we have found that achieved the best observed performance under passive decoding. The numerical simulation results are presented in Figure \ref{fig:BB passive decoding}. Unlike the case for ZSZ codes, we do not see reasonable evidence of a sustainable threshold. The left plot seems to suggest a transition near $p\approx 0.04\%$, but upon examination of the right plot, we see that this behavior is simply a finite-size effect: the logical error rates of the larger BB codes do not stabilize with increasing syndrome extraction cycles and eventually surpass those of the smaller BB codes.

\begin{table}[t]
\centering\renewcommand{\arraystretch}{1.5}
\begin{tabular}{@{}c|c|c|c|c|c@{}}
\hline
Decoding & Name & $\llbracket n,k,d \rrbracket$ & $\ell,m$  & $A$ & $B$  \\    
\hline
\multirow{3}{*}{passive}  
& BB144-2 & $\llbracket 144,12,\leq12 \rrbracket$ & 12,6 & $1+x^{11}y+x^4y^{5}$ & $x^6+x^7+x^8y^3$ \\
& BB360-2 & $\llbracket 360,12,\leq12 \rrbracket$ & 30,6 & $x^9y^2+x^{3}y^3+x^2y^{5}$ & $x^{25}+x^{21}+x^{18}y^5$ \\
& BB756-1 & $\llbracket 756,16,\leq20 \rrbracket$ & 21,18 & $1+x^{20}y^3+x^4y^{6}$ & $x^8y^9+x^{20}y^{17}+y^{17}$ \\
\hline
\end{tabular}
\caption{Bivariate bicycle codes and their parameters for the simulations in Figures \ref{fig:BB passive decoding} are displayed. Code distances are numerically estimated using the GAP package \textsf{QDistRnd} \cite{QDistRnd}.}
\label{tab:BB codes}
\end{table}

\begin{figure}[t]
\centering
\includegraphics[width=0.48\textwidth]{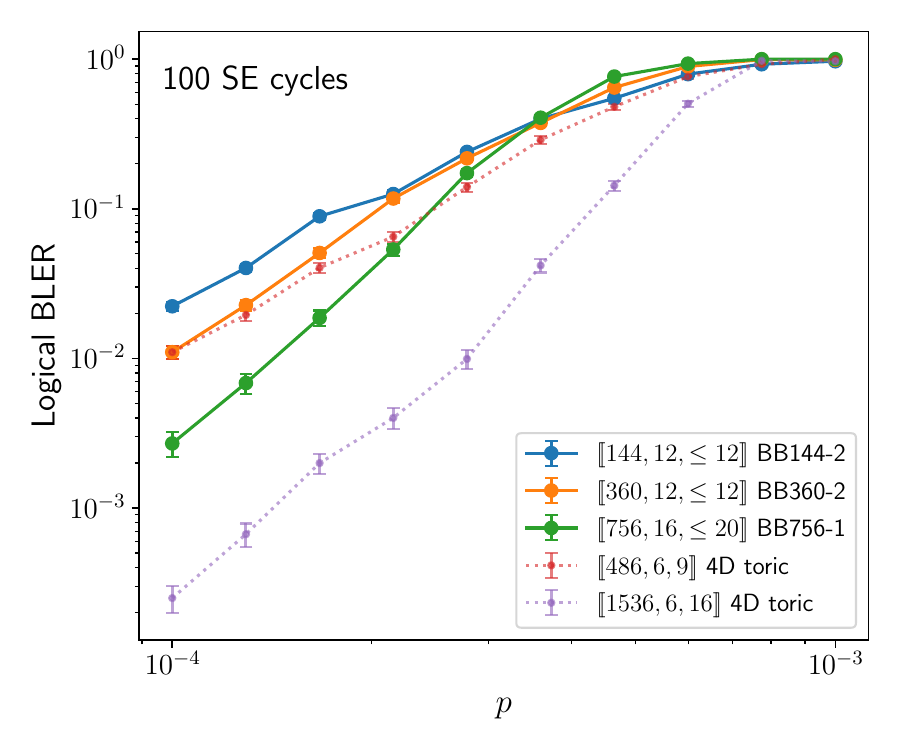}\hfill
\includegraphics[width=0.48\textwidth]{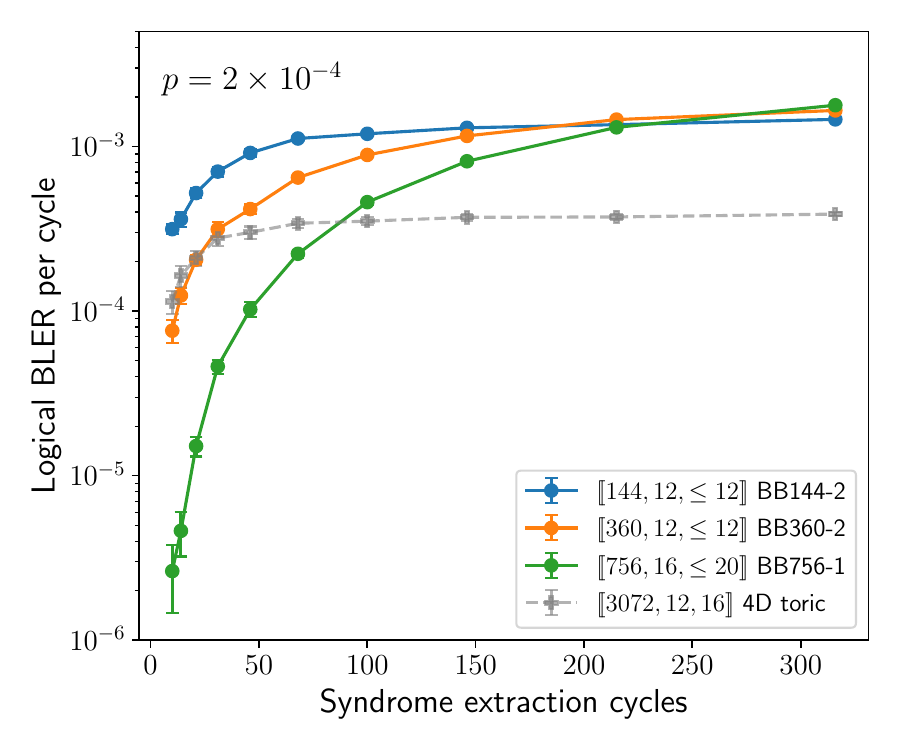}
\caption{The analogue of Figure \ref{fig:passive decoding} for the BB codes is displayed. In contrast to ZSZ codes, we do not observe any evidence of a sustainable threshold for BB codes under passive decoding.}
\label{fig:BB passive decoding}
\end{figure}

\bibliography{thebib}
\end{document}